\documentclass{article}

\usepackage{microtype}
\usepackage{graphicx}
\usepackage{subcaption}
\usepackage{booktabs} 
\usepackage{multirow}
\usepackage{hyperref}
\usepackage{float}        
\usepackage{placeins}

\usepackage[preprint]{icml2026}

\usepackage{amsmath}
\usepackage{amssymb}
\usepackage{mathtools}
\usepackage{amsthm}
\usepackage{enumitem}

\usepackage[capitalize,noabbrev]{cleveref}

\theoremstyle{plain}
\newtheorem{theorem}{Theorem}[section]
\newtheorem{proposition}[theorem]{Proposition}
\newtheorem{lemma}[theorem]{Lemma}

\theoremstyle{definition}
\newtheorem{definition}[theorem]{Definition}

\theoremstyle{remark}

\usepackage[textsize=tiny]{todonotes}
\usepackage{algorithmic}
\usepackage{algorithm}

\begin{document}

\twocolumn[
  \icmltitle{Information-Theoretic Privacy Control for Sequential Multi-Agent LLM Systems}



  \icmlsetsymbol{equal}{*}

  \begin{icmlauthorlist}
    \icmlauthor{Sadia Asif}{sch}
    \icmlauthor{Mohammad Mohammadi Amiri}{sch}
  \end{icmlauthorlist}

  \icmlaffiliation{sch}{Department of Computer Science, Rensselaer Polytechnic Institute, Troy New York}

  \icmlcorrespondingauthor{Sadia Asif}{asifs@rpi.edu}
  \icmlcorrespondingauthor{Mohammad Mohammadi Amiri}{ mamiri@rpi.edu}


  \vskip 0.3in
]




\printAffiliationsAndNotice{}

\begin{abstract}
Sequential multi-agent large language model (LLM) systems are increasingly
deployed in sensitive domains such as healthcare, finance, and enterprise
decision-making, where multiple specialized agents collaboratively process a
single user request. Although individual agents may satisfy local privacy
constraints, sensitive information can still be inferred through sequential
composition and intermediate representations. In this work, we study
\emph{compositional privacy leakage} in sequential LLM agent pipelines. We
formalize leakage using mutual information and derive a theoretical bound that
characterizes how locally introduced leakage can amplify across agents under
sequential execution. Motivated by this analysis, we propose a
privacy-regularized training framework that directly constrains information flow
between agent outputs and agent-local sensitive variables. We evaluate our
approach across sequential agent pipelines of varying depth on three benchmark
datasets, demonstrating stable optimization dynamics and consistent,
interpretable privacy-utility trade-offs. Our results show that privacy in
agentic LLM systems cannot be guaranteed by local constraints alone and must
instead be treated as a system-level property during both training and
deployment.
\end{abstract}

\section{Introduction}

Large language models (LLMs) are increasingly deployed as \emph{agentic systems},
where complex user requests are handled not by a single monolithic model, but by
multiple interacting agents with specialized roles.
Such multi-agent architectures enable modular reasoning, tool use, delegation,
and specialization, and are now widely adopted in applications spanning medical
decision support, financial analysis, and enterprise automation
\citep{pan2025swegym, xiao2024tradingagents, marreed2025enterprise}.

In such multi-agent systems, agents can interact in several structural modes, including
\emph{parallel} collaboration (agents operating independently and aggregating
results), \emph{hierarchical} orchestration (controllers delegating subtasks to
workers), and \emph{sequential} pipelines, where agents process information
stage-by-stage and pass intermediate representations downstream.
Each interaction pattern induces distinct failure modes and security risks.
In this work, we focus specifically on \emph{sequential multi-agent pipelines},
a common deployment pattern in real-world systems where intermediate reasoning
states, summaries, or tool outputs produced by one agent serve as input to the
next.

Despite their practical advantages, sequential agent pipelines introduce
\emph{privacy risks} that are not captured by traditional single-model threat
models.
In realistic deployments, individual agents often operate with
\emph{agent-local sensitive context}, such as proprietary policies, internal
routing logic, private memory states, or domain-specific confidential data.
Although downstream agents may not be granted direct access to this sensitive
information, intermediate messages and learned representations can retain
statistical dependencies that allow sensitive attributes to be inferred through
composition.
Consequently, even when each agent individually satisfies a local privacy
constraint, the system as a whole may exhibit substantial privacy leakage due to
sequential execution.

Prior research on privacy in LLMs has primarily focused on \emph{single-model}
settings, including memorization of training data
\citep{xu2025positional, nasr2025scalable},
membership inference attacks
\citep{meeus2025sok, mangaokar2025discrediting},
and differential privacy during training
\citep{kulynych2025unifying}.
While these approaches provide important guarantees at the model level, they do
not address privacy risks that arise \emph{at deployment time} from interactions
among multiple agents.
More recent work on multi-agent LLM security has examined adversarial prompt
propagation, trust miscalibration, and coordination failures
\citep{patil2025sum, zhang2025effective, yang2025agentnet,gomaa2025converse}.
However, privacy leakage induced by \emph{sequential composition of agent
representations} remains poorly understood and largely unquantified.
This gap is increasingly problematic as modern LLM deployments shift toward
agentic architectures, where intermediate representations, summaries, and tool
outputs are explicitly reused across multiple stages.
In such systems, privacy risk no longer stems solely from what a single agent
reveals, but from how information is \emph{propagated, transformed, and amplified} across agent boundaries.

Existing privacy techniques are ill-suited to capture this phenomenon.
Model-level guarantees (e.g., memorization bounds or differential privacy) treat
agents in isolation, while heuristic defenses such as sanitization or access
control reason only about surface-level disclosures.
Neither perspective accounts for \emph{latent statistical dependencies} embedded
in intermediate representations, nor for the cumulative effect of downstream
processing.
As a result, systems may satisfy strong local privacy constraints at each agent,
yet still exhibit substantial \emph{global leakage} once outputs are composed
sequentially.
Addressing this challenge requires a framework that (i) quantifies privacy at
the level of intermediate representations rather than final outputs, and (ii) provides
compositional guarantees that scale with pipeline depth. In this paper, we are addressing this gap by using \emph{information-theoretic perspective} to study
privacy leakage in sequential LLM agent pipelines.

\paragraph{Contributions.}
Our main contributions are summarized as follows:
\setlength{\itemsep}{1pt}
\setlength{\parskip}{1pt}
\setlength{\parsep}{1pt}
\begin{itemize}[topsep=2pt,itemsep=1pt,parsep=0pt,partopsep=0pt]
    \item We formalize \emph{compositional privacy leakage} in sequential
    multi-agent LLM systems and show that local per-agent privacy constraints
    are insufficient to guarantee global privacy.
    \item We derive a theoretical bound characterizing how information leakage
    amplifies across sequential agent pipelines under standard Markov
    assumptions.
    \item We propose a privacy-regularized training framework that directly
    constrains mutual information (MI) between intermediate agent representations
    and agent-local sensitive variables.
    \item We provide extensive empirical evaluation across medical, financial,
    and action-based privacy benchmarks, demonstrating consistent and
    interpretable privacy-utility trade-offs.
\end{itemize}

\section{Related Work}

A large body of work has studied privacy risks in LLMs, primarily
focusing on memorization-based attacks and training-data leakage.
Prior studies demonstrate that LLMs can regurgitate sensitive sequences from
their training data through extraction and membership inference attacks
\citep{shokri2017membership, carlini2022membership, ippolito2023false},
with memorization scaling with model size and data repetition.
More recent work has expanded this scope to inference-time privacy violations,
where models infer sensitive attributes from prompts or contextual cues even
when such information was not explicitly memorized
\citep{staab2024anonymizers,wu2025promptleakage}.
While effective at characterizing single-model risks, these approaches largely
assume centralized interaction and do not account for privacy leakage emerging
from multi-step or distributed reasoning.

In parallel, research on contextual privacy emphasizes that privacy is not an
absolute property but depends on socially appropriate information flows
\citep{bagdasarian2024airgapagent}.
Benchmarks such as PrivacyLens operationalize this view by evaluating whether
model actions violate contextual integrity under realistic scenarios
\citep{shao2024privacylens}.
However, these approaches evaluate privacy primarily at the level of final
actions or decisions, leaving open how intermediate representations contribute
to downstream leakage when composed across multiple agents.

The rise of multi-agent LLM systems further complicates this landscape \citep{zhang2025searching}.
Agentic frameworks decompose complex tasks into sequences of specialized agents,
often sharing intermediate representations, tool outputs, or memory states
\citep{talebirad2023multiagent, khan2025generative, prabhakar2025apigenmt}.
Although such decomposition improves task performance, recent studies reveal
that agent communication and memory mechanisms introduce new attack surfaces,
including memory poisoning and cross-agent inference
\citep{gomaa2025converse, rosser2025agentbreeder}.
Existing defenses focus on sanitization, prompt filtering, or architectural
isolation, but lack principled guarantees against leakage arising from sequential
composition.
Our work differs from prior literature in two key respects.
First, we study privacy leakage as a \emph{compositional phenomenon}, showing that
even benign, locally safe agent outputs can jointly reveal sensitive attributes
when aggregated.
Second, we provide an information-theoretic framework that quantifies this
leakage at the representation level and yields formal guarantees on cumulative
risk.
A more detailed discussion of memorization attacks, contextual integrity,
agent memory vulnerabilities, and their limitations is provided in
Appendix~\ref{app:related_work_extended}.
\section{Problem Setup}
We now formalize privacy leakage in sequential multi-agent LLM systems from an
information-theoretic perspective.
Our focus is on understanding how sensitive information held locally by
individual agents propagates through intermediate representations and
accumulates under sequential composition.
We begin by introducing the system model and threat assumptions, then define a
notion of compositional privacy leakage that captures system-level exposure
arising from multi-agent execution.
This formulation provides the foundation for our theoretical analysis of leakage
amplification and motivates the privacy-regularized training framework developed
in subsequent sections.

\subsection{Sequential Agent Pipeline}

We consider a sequential multi-agent system in which a single user request is
processed through a pipeline of $N$ specialized agents,
$\{a_1, \ldots, a_N\}$.
The request is first handled by agent $a_1$, whose output is passed to $a_2$, and
so on, until agent $a_N$ produces the final system response.

Let $O_0 \triangleq \varnothing$ denote an empty initial input.
For each agent $a_i$, we define:
\begin{itemize}[topsep=2pt,itemsep=1pt,parsep=0pt,partopsep=0pt]
    \item $D_i$ as public or task-relevant inputs to agent $a_i$,
    \item $S_i$ as \emph{agent-local sensitive information} accessible only to
    agent $a_i$ and $S_{<i} = \{S_1, \dots, S_{i-1}\}$,
    \item $O_i$ as the output representation produced by agent $a_i$.
\end{itemize}

Agent execution follows the sequential update rule
\begin{equation}
O_i = A_i(O_{i-1}, D_i, S_i), \qquad i = 1, \ldots, N,
\end{equation}
where $O_{i-1}$ denotes the output received from the preceding agent.
The final system output is given by $O_N$.

\vspace{4pt}
\paragraph{Conditional Independence Structure.}
We assume that each agent’s output depends on its local sensitive context only
through its own computation, and that downstream agents do not access upstream
sensitive variables directly.
Formally, the sequential pipeline satisfies the conditional Markov structure
\begin{equation}
\label{eq:markov}
S_{i-1} \rightarrow (O_{i-1}, D_i, S_i) \rightarrow O_i \rightarrow O_{i+1}
\rightarrow \cdots \rightarrow O_N ,
\end{equation}
for each $i \in \{1,\ldots,N\}$.
\paragraph{Interpretation.}
Eq. ~\eqref{eq:markov} formalizes the assumption that sensitive information
introduced at earlier stages influences downstream behavior only through the
intermediate representations that are explicitly passed along the pipeline.
Conditioned on the information available to agent $a_i$—namely the upstream
output $O_{i-1}$, task-relevant inputs $D_i$, and its own local sensitive context
$S_i$—the agent’s output $O_i$ is independent of upstream sensitive variables
$S_{i-1}$.
As a result, any leakage of upstream sensitive information into later outputs
must occur indirectly via sequential transformation and reuse of intermediate
representations, rather than through direct access to private agent state.
Although $S_i$ is never explicitly transmitted, statistical dependencies between
$S_i$ and the agent’s output $O_i$ can induce privacy leakage.
Because intermediate outputs are reused by downstream agents, such leakage may
propagate and accumulate through the pipeline via sequential transformations.

To quantify this phenomenon, we adopt an information-theoretic perspective and
measure privacy leakage using MI.
We define the \emph{global compositional leakage} as
\begin{equation}
\label{eq:global_leakage}
\mathcal{L}_{\text{global}}
\;\triangleq\;
I\!\left( O_N \,;\, S_1, \ldots, S_N \right),
\end{equation}
which captures how much information about all agent-local sensitive variables $\{S_i\}_{i=1}^N$ can
be inferred from the final system output $O_N$.
This definition naturally accounts for indirect leakage arising from sequential
composition, even when no individual agent directly discloses its sensitive
context.
\subsection{Threat Model}

We assume an adversary that observes the final output $O_N$ and seeks to infer
information about the sensitive variables $\{S_i\}_{i=1}^N$. The adversary does
not have direct access to agent internals, intermediate representations, or
private contexts, but can exploit statistical correlations induced by sequential
execution and representation passing.

This threat model captures both external observers accessing system outputs and
internal agents attempting to infer upstream sensitive information from shared
representations.
\subsection{Privacy Objective}

Our objective is to design learning and execution mechanisms that limit
compositional privacy leakage while preserving task utility.
Formally, we aim to control the global leakage
$\mathcal{L}_{\text{global}}$ subject to maintaining performance on the target
task.

Achieving this objective requires understanding how privacy loss arises from
the interaction of agents through shared intermediate representations, rather
than from any single agent in isolation.
In the next section, we formally analyze how privacy leakage behaves under
sequential composition and characterize the fundamental limits of agent-local
privacy control.

\section{Theoretical Analysis for Compositional Privacy Leakage}
\label{sec:theoretical-analysis}

We now provide a formal analysis of privacy leakage in sequential multi-agent
systems.
Our goal is to characterize how information about agent-local sensitive
variables propagates through intermediate representations and accumulates at
the system output.

We show that even when each agent individually satisfies a bounded local privacy
constraint, the sequential reuse and transformation of representations can
induce significantly larger leakage at the system level.
This theoretical analysis reveals a fundamental gap between agent-local privacy guarantees
and end-to-end privacy behavior in composed agent pipelines.

\subsection{Per-Agent Leakage}

For each agent $a_i$, we quantify \emph{local leakage} as the MI
between the agent’s output representation $O_i$ and its agent-local sensitive
variable $S_i$:
\begin{equation}
I(O_i; S_i).
\end{equation}
A natural design goal is to enforce a per-agent privacy constraint of the form
\begin{equation}
\label{eq:local_leakage}
I(O_i; S_i) \leq \epsilon_i,
\end{equation}
where $\epsilon_i \ge 0$ denotes a user-specified privacy budget for agent
$a_i$.

Such constraints ensure that no single agent output reveals excessive
information about its own sensitive context. However, as we show next, local
control alone is insufficient to bound the privacy risk of the composed system.
\subsection{Sequential Composition and Leakage Amplification}

Our analysis focuses on how the global compositional leakage defined in
Eq.~\eqref{eq:global_leakage} evolves under sequential execution.
Although each agent may satisfy a local information constraint, dependencies
introduced at early stages are repeatedly transformed by downstream agents,
leading to amplified exposure at the final output.
We show that this accumulation effect can dominate system-level behavior,
rendering agent-local guarantees insufficient in deep pipelines.

\vspace{4pt}
\paragraph{Local vs.\ Global Leakage.}
While per-agent leakage bounds defined in Eq.~\eqref{eq:local_leakage} control direct disclosure at individual stages,
they do not characterize how information propagates across the pipeline.
In sequential settings, correlations introduced upstream are transformed and
reused by downstream agents, allowing sensitive dependence to accumulate beyond
what any single agent reveals in isolation.
As a result, system-level privacy behavior can differ qualitatively from
agent-level guarantees.
\begin{theorem}[Cumulative Leakage Bound]
\label{thm:cumulative}
Consider a sequential agent pipeline satisfying the Markov structure as given in Eq.~\eqref{eq:markov}
and assume that the sensitive variables $\{S_i\}_{i=1}^N$ are mutually
independent. If each agent satisfies the local leakage constraint
$I(O_i; S_i) \leq \epsilon_i$, then the global compositional leakage satisfies
\begin{equation}
I(O_N; S_1, \ldots, S_N)
\;\leq\;
\sum_{i=1}^{N} 2^{\,N-i} \, \epsilon_i .
\end{equation}
\end{theorem}

\paragraph{Interpretation.}
Theorem~\ref{thm:cumulative} (proof  given in Appendix~\ref{app:proof_cumulative})
reveals an inherent \emph{amplification effect} under sequential composition:
leakage introduced by earlier agents is repeatedly transformed and propagated by
downstream agents, resulting in an exponentially growing contribution to the
final output.
In the uniform case $\epsilon_i = \epsilon$, the bound reduces to
$(2^{N}-1)\epsilon$, indicating that even modest per-agent leakage can yield
substantial global exposure as the pipeline depth increases.
Although this result is stated as an upper bound, our experimental results in
Section~\ref{sec:leakage-amplification} empirically exhibit the same depth-dependent
amplification trend, confirming that the bound captures a practically relevant
failure mode rather than a purely vacuous worst case.
\paragraph{Early-Agent Dominance.}
The structure of the bound implies that leakage introduced by earlier agents
contributes exponentially more to the final leakage than leakage introduced by
later agents. This asymmetry motivates stronger privacy control at early stages
of the pipeline and is empirically validated in our ablation studies provided in Appendix ~\ref{app:ablations}.
\begin{proposition}[Early-Agent Leakage Dominance]
\label{lem:early}
For fixed $\{\epsilon_i\}_{i=1}^N$, the marginal contribution of $\epsilon_i$ to
$I(O_N; S_1, \ldots, S_N)$ scales as $2^{N-i}$, implying that leakage at earlier
agents dominates global compositional leakage.
\end{proposition}

\subsection{Design Implications}

Theorem~\ref{thm:cumulative} demonstrates that enforcing independent,
per-agent privacy constraints does not suffice to control privacy risk in
sequential pipelines. Instead, privacy must be treated as a system-level
property that explicitly accounts for downstream amplification effects.
This observation motivates learning objectives that directly regulate
information flow throughout the pipeline, rather than relying on post hoc
filtering or isolated local defenses.

In the next section, we introduce a privacy-regularized training framework that
operationalizes this principle by explicitly penalizing MI
between agent outputs and agent-local sensitive variables during training.

\section{Privacy-Regularized Sequential Agent Training}

We now describe a training framework that explicitly controls privacy leakage in
sequential multi-agent systems. Motivated by the theoretical analysis of compositional privacy leakage
in Section~\ref{thm:cumulative}, our approach introduces an information-theoretic
regularization term that constrains the dependence between each agent’s output
and its agent-local sensitive context during training.

\subsection{Training Objective}

Our objective is to preserve task utility while limiting information leakage at
each stage of the agent pipeline. Let $y$ denote the ground-truth task label. The final agent output $O_N$ is used
to compute a standard supervised utility loss
\begin{equation}
\mathcal{L}_{\text{utility}} = \ell(O_N, y),
\end{equation}
where $\ell(\cdot)$ denotes a standard task loss (e.g., cross-entropy).

To control privacy leakage, we penalize the MI between each
agent’s output and its corresponding sensitive variable:
\begin{equation}
\mathcal{L}_{\text{privacy}}
=
\sum_{i=1}^{N} \beta_i \, I(O_i; S_i),
\end{equation}
where $\beta_i \ge 0$ determines the privacy-utility trade-off at agent $a_i$. The resulting optimization problem is
\begin{equation}
\label{eq:total_loss}
\mathcal{L}_{\text{total}}
=
\mathcal{L}_{\text{utility}}
+
\mathcal{L}_{\text{privacy}}.
\end{equation}

\subsection{Mutual Information Estimation}

Direct computation of the MI $I(O_i; S_i)$ is intractable for
high-dimensional neural representations produced by large language models.
We therefore adopt a variational estimation approach based on the
Donsker--Varadhan (DV) representation of the Kullback--Leibler (KL) divergence,
as instantiated by Mutual Information Neural Estimation (MINE)
\citep{belghazi2018mine}.

\paragraph{Donsker--Varadhan Representation.}
For two distributions $P$ and $Q$ defined on a common support $\Omega$, the
KL-divergence admits the variational form
\begin{equation}
D_{\text{KL}}(P \,\|\, Q)
=
\sup_{T:\Omega \times \Omega \rightarrow \mathbb{R}}
\left\{
\mathbb{E}_{P}[T]
-
\log \mathbb{E}_{Q}\!\left[e^{T}\right]
\right\},
\label{eq:dv-representation}
\end{equation}
where the supremum is taken over all measurable functions $T$ for which the
expectations are finite.

By setting $P = p(o_i, s_i)$ and $Q = p(o_i)p(s_i)$, the MI
between an agent’s output $O_i$ and its sensitive variable $S_i$ can be written
as
\begin{align}
I(O_i; S_i)
&= D_{\text{KL}}\!\big(p(o_i, s_i) \,\|\, p(o_i)p(s_i)\big) \nonumber \\
&= \sup_{T \in \mathcal{F}} \Biggl\{
\mathbb{E}_{p(o_i,s_i)}[T(o_i,s_i)] \nonumber \\
&\quad - \log \mathbb{E}_{p(o_i)p(s_i)}\big[e^{T(o_i,s_i)}\big]
\Biggr\},
\label{eq:mine-objective}
\end{align}

where $\mathcal{F}$ denotes a parameterized function family.

\paragraph{Neural Estimator Formulation.}
In practice, we restrict $\mathcal{F}$ to a family of neural networks
$T_{\psi_i} : \mathcal{O}_i \times \mathcal{S}_i \rightarrow \mathbb{R}$,
where $\mathcal{O}_i$ denotes the intermediate representation produced by agent $a_i$ and $\mathcal{S}_i$ denotes the corresponding agent-local sensitive variable,
referred to as \emph{critics}, with parameters $\psi_i$.
Given a minibatch of size $B$, we obtain empirical joint samples
$\{(o_i^{(b)}, s_i^{(b)})\}_{b=1}^{B}$ from $p(o_i,s_i)$, and approximate samples
from the product of marginals $p(o_i)p(s_i)$ by pairing $o_i^{(b)}$ with
shuffled sensitive variables $s_i'^{(b)}$ within the minibatch.

The resulting sample-based MINE estimator is

\begin{equation}
\begin{aligned}
\hat{I}_{\text{MINE}}(O_i; S_i)
&= \frac{1}{B} \sum_{b=1}^{B} T_{\psi_i}(o_i^{(b)}, s_i^{(b)}) \\
&\hspace*{-1.5em}-\log \Bigg(
\frac{1}{B} \sum_{b=1}^{B} \exp\big(T_{\psi_i}(o_i^{(b)}, s_i'^{(b)})\big)
\Bigg).
\end{aligned}
\label{eq:mine-estimate}
\end{equation}

This estimator provides a differentiable bound on the true mutual
information and serves as a quantitative measure of privacy leakage at agent
$a_i$.
\paragraph{Optimization and Stability.}
During training, the critic parameters $\psi_i$ are updated to 
\emph{maximize} $\hat{I}_{\text{MINE}}(O_i; S_i)$ in Eq. ~\eqref{eq:mine-estimate}, while the agent parameters are updated to \emph{minimize} the same 
quantity as part of the overall privacy-regularized objective. 
To improve numerical stability of the logarithmic partition term, we maintain 
an exponential moving average of $\exp(T_{\psi_i})$, following standard and 
empirically validated MINE implementations \citep{belghazi2018mine}.

\subsection{Optimization Procedure}

Training proceeds iteratively and alternates between estimating mutual
information and updating the agent parameters to minize $\mathcal{L}_{\text{total}}$ given in Eq.~\eqref{eq:total_loss}. Each iteration consists of three
phases: (i) a forward pass through the agent pipeline to obtain  $\mathcal{L}_{\text{utility}}$,
(ii) updating the MINE critics to obtain accurate estimates of 
$\hat{I}_{\text{MINE}}(O_i; S_i)$, and
(iii) updating the agent parameters by minimizing the combined objective
$\mathcal{L}_{\text{total}}$, which balances task utility and privacy
regularization.
The complete training procedure is summarized in
Algorithm~\ref{alg:mine-training}.

\begin{algorithm}
\caption{Information-Regularized Multi-Agent Training with MINE}
\label{alg:mine-training}
\begin{algorithmic}[1]
\STATE Initialize agent parameters $\{\theta_i\}_{i=1}^{N}$ and estimator parameters $\{\psi_i\}_{i=1}^{N}$.
\FOR{iteration $t = 1$ to $T$}
    \STATE \textbf{Phase 1: Forward Pass}
    \STATE Compute sequential outputs: $O_1 \to O_2 \to \dots \to O_N$.
    \STATE Compute utility loss: $\mathcal{L}_{\text{utility}} = \ell(O_N, y)$.

    \STATE \textbf{Phase 2: MINE Estimator Update}
    \FOR{each agent $a_i$}
        \STATE Sample joint pairs $(o_i, s_i) \sim p(o_i, s_i)$.
        \STATE Sample marginal pairs $(o_i, s_i')$ by shuffling $S_i$.
        \STATE Estimate $\hat{I}_{\text{MINE}}(O_i; S_i)$ using Eq. \eqref{eq:mine-estimate}.
        \STATE Update $\psi_i \leftarrow \psi_i + \eta_\psi \nabla_{\psi_i} \hat{I}_{\text{MINE}}(O_i; S_i)$.
    \ENDFOR

    \STATE \textbf{Phase 3: Agent Update}
    \STATE Compute total loss:\\ $\mathcal{L}_{\text{total}} = \mathcal{L}_{\text{utility}} + \sum_i \beta_i \hat{I}_{\text{MINE}}(O_i; S_i)$.
    \STATE Update agent parameters: $\theta_i \leftarrow \theta_i - \eta_\theta \nabla_{\theta_i} \mathcal{L}_{\text{total}}$.
\ENDFOR
\end{algorithmic}
\end{algorithm}

This optimization enforces an explicit information bottleneck at each agent
boundary. As a result, agents learn representations that are sufficient for task
completion while minimizing the propagation of agent-local sensitive context
through the sequential pipeline.

\section{Experimental Setup}

We evaluate our privacy-regularized training framework on sequential
multi-agent LLM systems across multiple domains.
Our experiments are designed to isolate how privacy leakage and task utility
scale with agent depth, model capacity, and task domain.

\paragraph{Benchmarks.}
We consider three privacy-critical benchmarks:
\emph{MedQA} \citep{jin2020disease} for medical reasoning,
\emph{FinQA} \citep{chen2021finqa} for financial numerical reasoning,
and \emph{PrivacyLens} \citep{shao2024privacylens} for action-based evaluation
of contextual privacy norms.
Together, these benchmarks cover classification-style reasoning,
numerical program execution, and realistic agent trajectories where privacy
violations may emerge only through sequential composition.

\paragraph{Sequential Agent Pipelines.}
We evaluate pipelines of $N\in\{2,3,4,5\}$ agents executed sequentially.
Each agent is an independent LLM that communicates exclusively via intermediate
representations, reflecting common real-world deployments where reasoning,
planning, and policy enforcement are handled by distinct components.
\paragraph{Models.}
Experiments are conducted using open-weight LLMs from the Qwen (2B, 4B) and
LLaMA (3B, 7B) families,\footnote{
Qwen models: \url{https://huggingface.co/Qwen/Qwen2-1.5B},
\url{https://huggingface.co/Qwen/Qwen1.5-4B-Chat}.
LLaMA models: \url{https://huggingface.co/meta-llama/Llama-3.2-3B},
\url{https://huggingface.co/meta-llama/Llama-2-7b}.
}
fine-tuned with LoRA adapters \citep{hu2021lora}.
Agents do not share parameters unless explicitly stated.

\paragraph{Training and Baselines.}
We compare standard end-to-end training against privacy-regularized training
that penalizes MI between intermediate agent representations
and agent-local sensitive variables.
Details of the objective and sensitive variable definitions are
provided in Appendix~\ref{app:exp_details}.
\begin{table}[t]
\centering
\caption{MedQA results for LLaMA models (B=Baseline, M=MINE-Reg); Ag = number of agents.}
\label{tab:llama_models}
\scriptsize
\begin{tabular}{@{}l @{} c @{} c c c c c c c @{}}
\toprule
\textbf{Model} & \textbf{Ag}\kern3pt & \textbf{Meth} &
\textbf{CE} $\downarrow$ &
\textbf{MI$_{\text{avg}}$} $\downarrow$ &
\textbf{SB} $\uparrow$ &
\textbf{BS} $\uparrow$ &
\textbf{PI} $\uparrow$ &
\textbf{PARI} $\uparrow$ \\

\midrule
\multirow{8}{*}{LLaMA-7B}
& 2 & B & 1.32 & 0.49 & 0.42 & 0.95 & 0.000 & 0.480 \\
& 2 & M & 1.47 & 0.06 & 0.81 & 0.89 & 0.878 & 0.901 \\
& 3 & B & 1.54 & 0.70 & 0.33 & 0.93 & 0.000 & 0.474 \\
& 3 & M & 1.68 & 0.08 & 0.78 & 0.87 & 0.886 & 0.899 \\
& 4 & B & 1.73 & 0.90 & 0.26 & 0.90 & 0.000 & 0.465 \\
& 4 & M & 1.92 & 0.10 & 0.74 & 0.84 & 0.889 & 0.891 \\
& 5 & B & 1.88 & 1.05 & 0.21 & 0.87 & 0.000 & 0.456 \\
& 5 & M & 2.06 & 0.14 & 0.70 & 0.81 & 0.867 & 0.871 \\
\midrule

\multirow{8}{*}{LLaMA-3B}
& 2 & B & 1.55 & 0.62 & 0.34 & 0.86 & 0.000 & 0.453 \\
& 2 & M & 1.72 & 0.18 & 0.66 & 0.78 & 0.710 & 0.784 \\
& 3 & B & 1.71 & 0.80 & 0.29 & 0.83 & 0.000 & 0.444 \\
& 3 & M & 1.89 & 0.22 & 0.62 & 0.77 & 0.725 & 0.788 \\
& 4 & B & 1.88 & 0.96 & 0.25 & 0.80 & 0.000 & 0.435 \\
& 4 & M & 2.05 & 0.27 & 0.58 & 0.74 & 0.719 & 0.776 \\
& 5 & B & 2.02 & 1.12 & 0.20 & 0.77 & 0.000 & 0.426 \\
& 5 & M & 2.21 & 0.33 & 0.54 & 0.71 & 0.705 & 0.761 \\
\bottomrule
\end{tabular}
\end{table}
\paragraph{Evaluation.}
Utility is measured using task-appropriate accuracy or helpfulness metrics,
while privacy is evaluated using MI estimates, adversarial
leakage accuracy, and composite privacy-utility scores.
Formal metric definitions are given in Appendix~\ref{app:metrics}.
\begin{figure}[t]
    \centering
    \includegraphics[width=0.9\columnwidth]{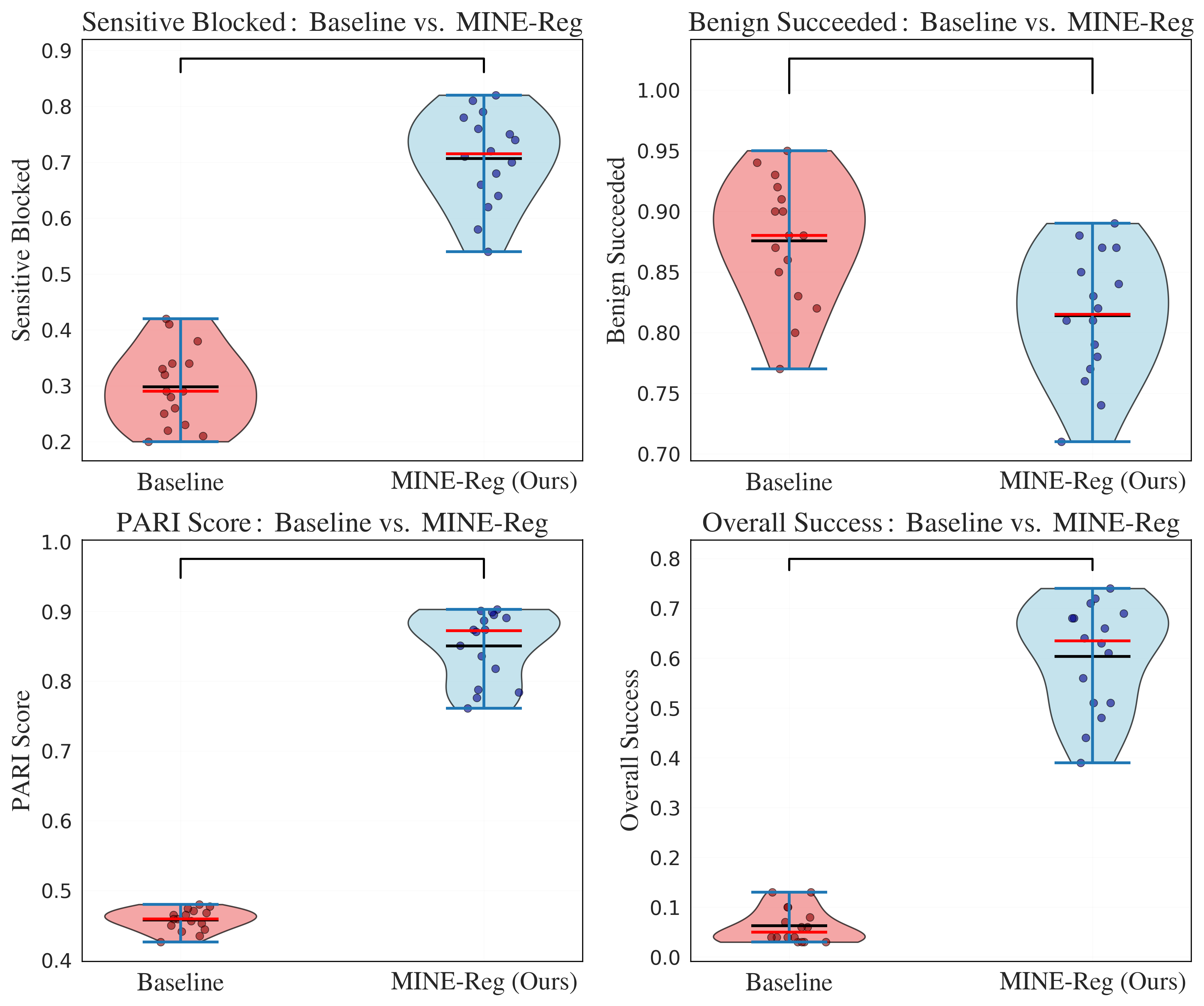}
    \caption{
Distributional comparison of baseline and MINE-Reg across privacy and utility metrics for LLaMA-3B.
Violin plots show per-run variability with mean (red bar) and dispersion for Sensitive Blocked, Benign Succeeded, PARI score, and Overall Success.
    }
    \label{fig:privacy_utility_tradeoff}
\end{figure}
\section{Results}
\label{sec:results}
In this section, we provide the results of our proposed information-theoretic regularization framework on
\textbf{medical} (MedQA) \citep{jin2020disease} and \textbf{financial} (FinQA) \citep{chen2021finqa} reasoning benchmarks,
focusing on privacy leakage, task utility, and their interaction under
sequential multi-agent composition.
All metrics follow the formal definitions given in Appendix~\ref{app:metrics}.
Results on the PrivacyLens \citep{shao2024privacylens}  benchmark are deferred to
Appendix~\ref{app:add-exp}.
\subsection{Overall Performance on Medical and Financial Benchmarks}
Tables \ref{tab:llama_models}, ~\ref{tab:qwen_models} report MedQA results, while
Tables \ref{tab:finqa_llama_models},~\ref{tab:finqa_qwen_models} report FinQA results.
Across both domains and all agent depths ($2$--$5$ agents), MI-regularized training
(MINE-Reg) consistently reduces compositional leakage and improves system-level
privacy outcomes, while preserving the majority of benign task performance.
Two consistent patterns emerge from the tables: (i) unregularized baselines exhibit
progressively worse privacy with increasing depth, and (ii) MINE-Reg flattens this
depth-induced degradation, yielding substantially higher privacy protection at comparable utility.
\begin{table}[t]
\centering
\caption{MedQA results for Qwen models (B=Baseline, M=MINE-Reg); Ag = number of agents.}
\label{tab:qwen_models}
\scriptsize

\begin{tabular}{@{}l @{} c @{} c c c c c c c @{}}
\toprule
\textbf{Model} & \textbf{Ag}\kern3pt & \textbf{Meth} &
\textbf{CE} $\downarrow$ &
\textbf{MI$_{\text{avg}}$} $\downarrow$ &
\textbf{SB} $\uparrow$ &
\textbf{BS} $\uparrow$ &
\textbf{PI} $\uparrow$ &
\textbf{PARI} $\uparrow$ \\

\midrule
\multirow{8}{*}{Qwen-4B}
& 2 & B & 1.40 & 0.54 & 0.38 & 0.94 & 0.000 & 0.477 \\
& 2 & M & 1.58 & 0.06 & 0.82 & 0.88 & 0.889 & 0.903 \\
& 3 & B & 1.62 & 0.73 & 0.32 & 0.92 & 0.000 & 0.471 \\
& 3 & M & 1.77 & 0.08 & 0.79 & 0.85 & 0.890 & 0.895 \\
& 4 & B & 1.75 & 0.88 & 0.28 & 0.90 & 0.000 & 0.465 \\
& 4 & M & 1.93 & 0.10 & 0.75 & 0.83 & 0.886 & 0.887 \\
& 5 & B & 1.89 & 1.02 & 0.22 & 0.88 & 0.000 & 0.459 \\
& 5 & M & 2.05 & 0.13 & 0.71 & 0.81 & 0.873 & 0.874 \\
\midrule

\multirow{8}{*}{Qwen-2B}
& 2 & B & 1.48 & 0.61 & 0.41 & 0.91 & 0.000 & 0.468 \\
& 2 & M & 1.55 & 0.10 & 0.76 & 0.87 & 0.836 & 0.874 \\
& 3 & B & 1.66 & 0.78 & 0.34 & 0.88 & 0.000 & 0.459 \\
& 3 & M & 1.74 & 0.14 & 0.72 & 0.82 & 0.821 & 0.851 \\
& 4 & B & 1.82 & 0.94 & 0.29 & 0.85 & 0.000 & 0.450 \\
& 4 & M & 1.93 & 0.18 & 0.68 & 0.79 & 0.809 & 0.836 \\
& 5 & B & 1.98 & 1.10 & 0.23 & 0.82 & 0.000 & 0.441 \\
& 5 & M & 2.10 & 0.23 & 0.64 & 0.76 & 0.791 & 0.818 \\
\bottomrule
\end{tabular}
\end{table}

\begin{table}[t]
\centering
\caption{FinQA results for LLaMA models (B=Baseline, M=MINE-Reg); Ag = number of agents.}
\label{tab:finqa_llama_models}
\scriptsize
\begin{tabular}{@{}l @{} c @{} c c c c c c c @{}}
\toprule
\textbf{Model} & \textbf{Ag}\kern3pt & \textbf{Meth} &
\textbf{CE} $\downarrow$ &
\textbf{MI$_{\text{avg}}$} $\downarrow$ &
\textbf{SB} $\uparrow$ &
\textbf{BS} $\uparrow$ &
\textbf{PI} $\uparrow$ &
\textbf{PARI} $\uparrow$ \\
\midrule
\multirow{8}{*}{LLaMA-7B}
& 2 & B & 1.38 & 0.51 & 0.40 & 0.94 & 0.000 & 0.481 \\
& 2 & M & 1.53 & 0.06 & 0.82 & 0.88 & 0.882 & 0.905 \\
& 3 & B & 1.61 & 0.72 & 0.31 & 0.90 & 0.000 & 0.472 \\
& 3 & M & 1.78 & 0.09 & 0.77 & 0.84 & 0.875 & 0.890 \\
& 4 & B & 1.79 & 0.93 & 0.25 & 0.89 & 0.000 & 0.466 \\
& 4 & M & 1.98 & 0.11 & 0.73 & 0.83 & 0.882 & 0.892 \\
& 5 & B & 1.95 & 1.10 & 0.20 & 0.86 & 0.000 & 0.458 \\
& 5 & M & 2.14 & 0.15 & 0.69 & 0.80 & 0.864 & 0.872 \\
\midrule

\multirow{8}{*}{LLaMA-3B}
& 2 & B & 1.62 & 0.69 & 0.31 & 0.85 & 0.000 & 0.449 \\
& 2 & M & 1.78 & 0.21 & 0.69 & 0.77 & 0.696 & 0.778 \\
& 3 & B & 1.79 & 0.87 & 0.26 & 0.82 & 0.000 & 0.441 \\
& 3 & M & 1.96 & 0.26 & 0.64 & 0.75 & 0.701 & 0.781 \\
& 4 & B & 1.95 & 1.02 & 0.22 & 0.79 & 0.000 & 0.433 \\
& 4 & M & 2.13 & 0.31 & 0.60 & 0.72 & 0.696 & 0.768 \\
& 5 & B & 2.10 & 1.18 & 0.18 & 0.76 & 0.000 & 0.425 \\
& 5 & M & 2.28 & 0.37 & 0.56 & 0.69 & 0.686 & 0.755 \\
\bottomrule
\end{tabular}
\end{table}

\begin{table}[t]
\centering
\caption{FinQA results for Qwen models (B=Baseline, M=MINE-Reg); Ag = number of agents.}
\label{tab:finqa_qwen_models}
\scriptsize
\begin{tabular}{@{}l @{} c @{} c c c c c c c @{}}
\toprule
\textbf{Model} & \textbf{Ag}\kern3pt & \textbf{Meth} &
\textbf{CE} $\downarrow$ &
\textbf{MI$_{\text{avg}}$} $\downarrow$ &
\textbf{SB} $\uparrow$ &
\textbf{BS} $\uparrow$ &
\textbf{PI} $\uparrow$ &
\textbf{PARI} $\uparrow$ \\
\midrule
\multirow{8}{*}{Qwen-4B}
& 2 & B & 1.46 & 0.57 & 0.36 & 0.93 & 0.000 & 0.479 \\
& 2 & M & 1.63 & 0.07 & 0.80 & 0.87 & 0.877 & 0.900 \\
& 3 & B & 1.69 & 0.75 & 0.31 & 0.91 & 0.000 & 0.473 \\
& 3 & M & 1.84 & 0.09 & 0.77 & 0.84 & 0.880 & 0.892 \\
& 4 & B & 1.82 & 0.92 & 0.27 & 0.88 & 0.000 & 0.464 \\
& 4 & M & 1.99 & 0.12 & 0.73 & 0.81 & 0.870 & 0.878 \\
& 5 & B & 1.96 & 1.08 & 0.21 & 0.85 & 0.000 & 0.455 \\
& 5 & M & 2.15 & 0.16 & 0.69 & 0.79 & 0.852 & 0.863 \\
\midrule

\multirow{8}{*}{Qwen-2B}
& 2 & B & 1.54 & 0.64 & 0.34 & 0.90 & 0.000 & 0.468 \\
& 2 & M & 1.62 & 0.12 & 0.76 & 0.85 & 0.812 & 0.858 \\
& 3 & B & 1.71 & 0.82 & 0.28 & 0.87 & 0.000 & 0.459 \\
& 3 & M & 1.80 & 0.16 & 0.72 & 0.80 & 0.805 & 0.842 \\
& 4 & B & 1.88 & 0.98 & 0.24 & 0.84 & 0.000 & 0.451 \\
& 4 & M & 1.98 & 0.21 & 0.68 & 0.77 & 0.786 & 0.827 \\
& 5 & B & 2.03 & 1.15 & 0.19 & 0.81 & 0.000 & 0.442 \\
& 5 & M & 2.15 & 0.26 & 0.64 & 0.74 & 0.774 & 0.812 \\
\bottomrule
\end{tabular}
\end{table}

MINE-Reg achieves large reductions in $\mathrm{MI}_{\text{avg}}$ relative to the baseline,
typically in the \textbf{75--90\%} range depending on depth and model scale.
On the MedQA benchmark, for Qwen-4B with two agents, $\mathrm{MI}_{\text{avg}}$ drops from $0.54$ to $0.06$
($\approx 89\%$ reduction), and remains low even as the pipeline deepens (e.g., $0.13$ at five agents
versus a baseline of $1.02$).
Similarly on the FinQA benchmark, LLaMA-7B reduces $\mathrm{MI}_{\text{avg}}$ from $0.51$ to $0.06$ at two agents
($\approx 88\%$), and from $1.10$ to $0.15$ at five agents ($\approx 86\%$).
These reductions translate into consistent gains in adversarial failure rates:
Sensitive Blocked (SB) increases by roughly \textbf{+0.35 to +0.50} across most configurations.
For example, on MedQA with Qwen-4B, SB improves from $0.38 \rightarrow 0.82$ at two agents and
from $0.22 \rightarrow 0.71$ at five agents; on FinQA with LLaMA-7B, SB improves from
$0.40 \rightarrow 0.82$ at two agents and from $0.20 \rightarrow 0.69$ at five agents.
Notably, the privacy integrity (PI) component mirrors these trends, rising from $0$ for baselines
(by definition) to $\approx 0.69$--$0.89$ under MINE-Reg, indicating substantial suppression
of sensitive dependence.

\paragraph{Utility Preservation.}
Despite strong privacy improvements, benign accuracy remains comparatively stable under
regularization. Across MedQA, Benign Succeeded (BS) typically decreases by \textbf{6--8 points} for larger models
(e.g., Qwen-4B: $0.94 \rightarrow 0.88$ at two agents; LLaMA-7B: $0.95 \rightarrow 0.89$),
and by \textbf{8--10 points} for smaller models (e.g., LLaMA-3B: $0.86 \rightarrow 0.78$).
On FinQA, benign success is similarly preserved, with LLaMA-7B decreasing from
$0.94 \rightarrow 0.88$ at two agents and $0.86 \rightarrow 0.80$ at five agents.
Cross-entropy (CE) increases modestly under MINE-Reg (reflecting the expected
privacy-utility trade-off), but the degradation remains bounded and does not eliminate the gains
in system-level reliability.

\paragraph{Privacy-Aware Reasoning Index (PARI).}
The composite metric PARI improves substantially across all settings, reflecting simultaneous gains
in privacy integrity and preserved reasoning correctness.
For Qwen-4B on MedQA, PARI increases from $0.477 \rightarrow 0.903$ at two agents and from
$0.459 \rightarrow 0.874$ at five agents; for LLaMA-7B on FinQA, PARI increases from
$0.481 \rightarrow 0.905$ at two agents and from $0.458 \rightarrow 0.872$ at five agents.
These improvements indicate that MI-regularization yields a consistently better operating point for
privacy-sensitive sequential deployments, particularly as agent depth increases.

\subsection{Privacy-Utility Trade-off Analysis}

Figure~\ref{fig:privacy_utility_tradeoff} provides a distributional comparison
between unregularized baselines and MINE-Reg for Llama-3B across varying agent depths on the MedQA benchmark.
Rather than emphasizing individual configurations, this analysis highlights
system-level trends aggregated across experimental settings.

\paragraph{Privacy Gains.}
The top-left panel shows the distribution of SB.
Baseline systems cluster in the low-privacy regime
($\text{SB} \approx 0.20$--$0.40$), indicating widespread vulnerability under
sequential composition.
In contrast, MINE-Reg shifts the distribution upward, with median SB near
$0.70$ and upper quartiles exceeding $0.80$, demonstrating consistent privacy
improvements across models and depths.

\paragraph{Utility Preservation.}
The top-right panel reports BS.
While MINE-Reg introduces a modest utility reduction, performance remains
concentrated in the high-accuracy regime
($\text{BS} \approx 0.78$--$0.88$).
The substantial overlap with baseline distributions confirms that privacy gains
are not achieved via catastrophic utility loss, consistent with the $5$--$9$
point BS drops observed in Tables ~\ref{tab:llama_models},~\ref{tab:qwen_models},\ref{tab:finqa_llama_models} and ~\ref{tab:finqa_qwen_models}.

\paragraph{Composite System Performance.}
The bottom-left panel reports the PARI.
Baseline systems cluster around $0.44$--$0.48$, reflecting persistent leakage
despite strong benign accuracy.
MINE-Reg shifts PARI sharply upward, with medians near $0.87$--$0.90$,
indicating a strictly better privacy-utility operating point.

\paragraph{Overall Success (OS).}
Similarly, the bottom-right panel shows OS, which requires both
correct benign behavior and successful privacy protection.
Baseline systems fail almost universally ($\text{OS} < 0.10$), whereas
MINE-Reg concentrates mass between $0.55$ and $0.70$. This sharp separation indicates that privacy regularization not only mitigates
leakage but also enables reliable end-to-end task completion under safety constraints.

\subsection{Leakage Amplification with Agent Depth}
\label{sec:leakage-amplification}
\begin{figure}[t]
    \centering
    \includegraphics[width=0.9\columnwidth]{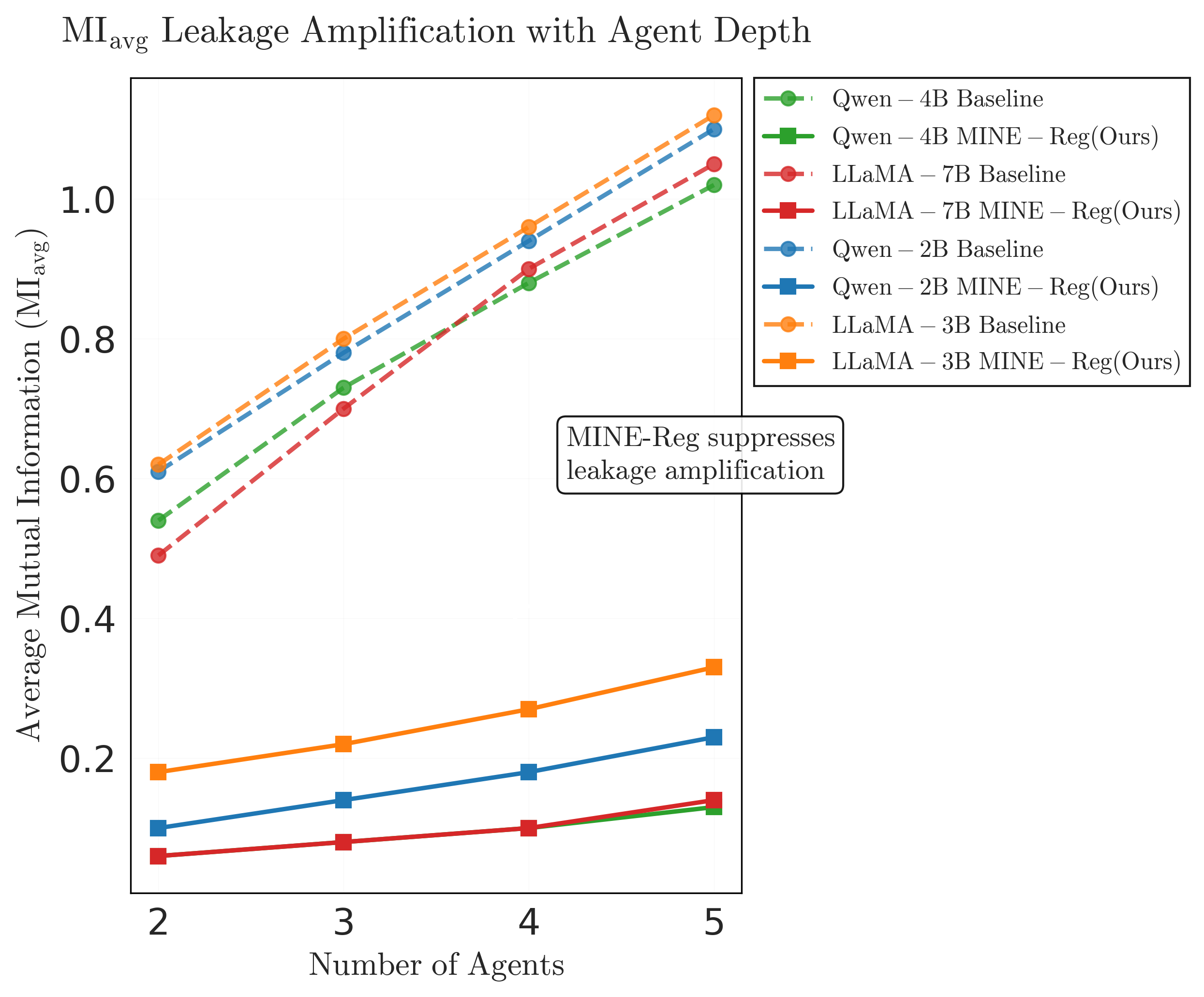}
    \caption{
    $\mathrm{MI}_{\text{avg}}$  leakage as a function of sequential agent depth across different models for MedQA benchmark. Unregularized systems exhibit strong leakage amplification with depth,
    while MI-regularized training effectively suppresses cumulative leakage
    across all model families.
    }
    \label{fig:mi_leakage_depth}
\end{figure}
Figure~\ref{fig:mi_leakage_depth} illustrates how representation-level $\mathrm{MI}_{\text{avg}}$ evolves with sequential agent depth for Qwen-2B, Qwen-4B,
LLaMA-3B, and LLaMA-7B models.
Across all model families, unregularized baselines exhibit clear leakage
amplification: $\mathrm{MI}_{\text{avg}}$ increases monotonically as the number of agents grows from
two to five, consistent with compounding disclosure effects predicted by
Theorem~\ref{thm:cumulative}.
For example, on MedQA with LLaMA-7B, $\mathrm{MI}_{\text{avg}}$ rises from $0.49$ at two agents to $1.05$
at five agents, with similar trends observed for Qwen variants.

In contrast, MI-regularized training substantially suppresses this amplification
across all models.
Under MINE-Reg, $\mathrm{MI}_{\text{avg}}$ remains bounded and grows only mildly with depth, staying
below approximately $0.35$ even at five agents.
This separation between baseline and regularized curves highlights that local
per-agent constraints alone are insufficient in deep pipelines, and that
explicit system-level MI regularization is necessary to control cumulative
privacy leakage under sequential composition.

\begin{figure}[t]
    \centering
    \includegraphics[width=0.9\columnwidth]{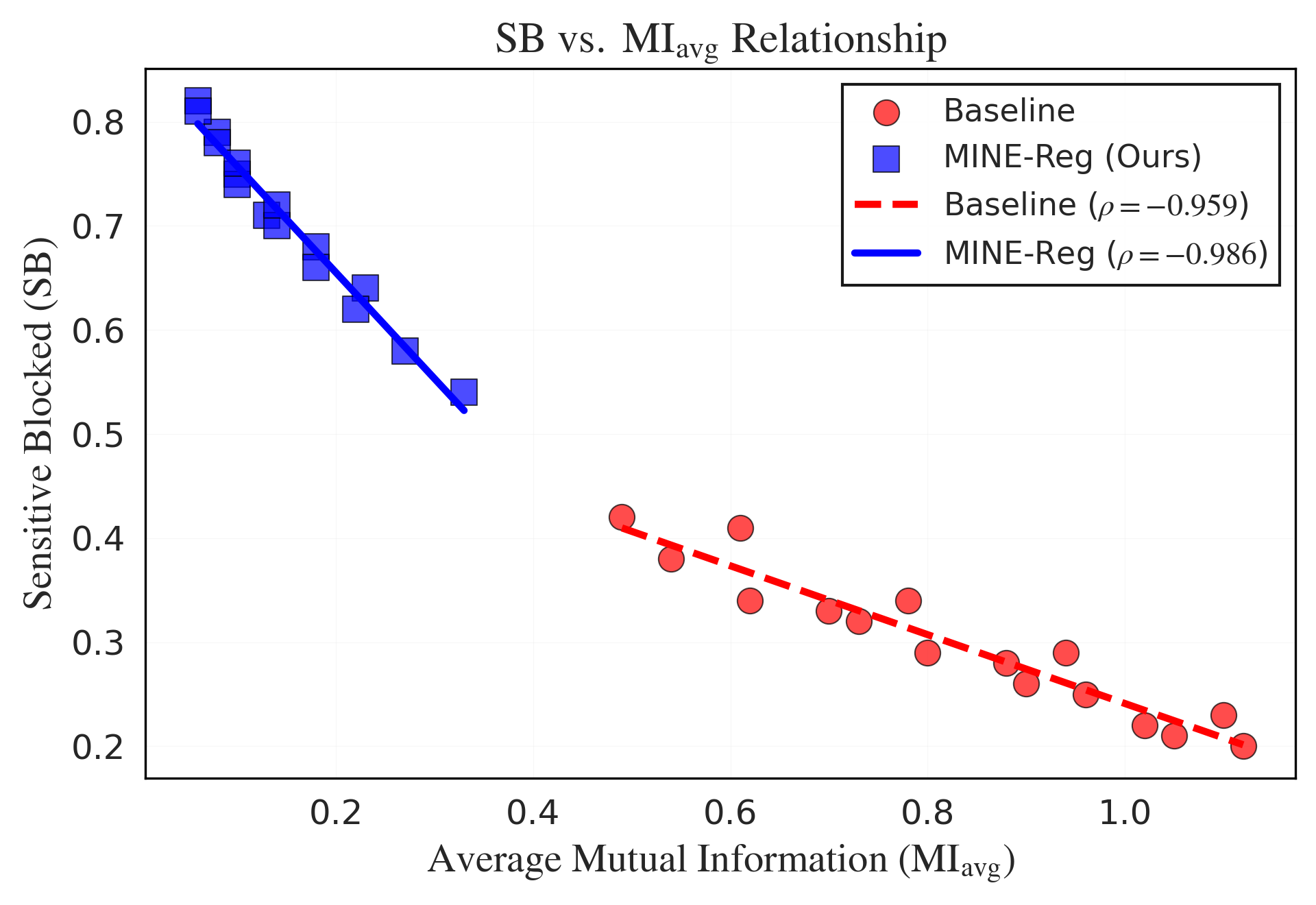}
    \caption{
    Relationship between $\mathrm{MI}_{\text{avg}}$ and
    SB for Qwen-4B model on FinQA benchmark.
    Strong negative correlations validate the information-theoretic formulation:
    reducing MI directly improves privacy, with steeper gains under
    MI-regularized training.
    }
    \label{fig:sb_vs_mi}
\end{figure}

\subsection{Information-Theoretic Validation}

Figure \ref{fig:sb_vs_mi} empirically validates the information-theoretic foundation of our approach by showing a strong negative relationship between $\mathrm{MI}_{\text{avg}}$ and SB for the Qwen-4B model on the FinQA benchmark. Both unregularized baselines and MINE-Reg exhibit a strong negative correlation
between $\mathrm{MI}_{\text{avg}}$ and SB, with Pearson coefficients of $\rho=-0.959$ and $\rho=-0.986$,
respectively.
This near-linear dependence confirms that reducing $\mathrm{MI}_{\text{avg}}$ directly improves privacy
outcomes, supporting it as a reliable proxy for leakage risk.
Notably, the steeper slope under MI-regularization indicates heightened
sensitivity in low-leakage regimes, where small reductions in $\mathrm{MI}_{\text{avg}}$ yield
disproportionately large gains in privacy protection.These findings empirically support the cumulative leakage analysis in
Section~\ref{sec:theoretical-analysis}, demonstrating that controlling representation-level MI
effectively bounds downstream privacy risk.

\section{Conclusion and Future Work}
We analyzed privacy leakage in sequential multi-agent LLM pipelines and showed that local privacy guarantees fail to compose, leading to amplified global leakage. To address this, we introduced MI-regularized training (MINE-Reg), a system-level approach that directly constrains information flow at intermediate representations. Both theoretical analysis and experiments on medical and financial benchmarks demonstrate that MINE-Reg effectively suppresses cumulative leakage while preserving task utility, yielding large gains in overall system reliability.

Several directions remain open. Extending the framework to dynamic or adaptive agent pipelines would better reflect real-world deployments. Incorporating leakage-aware orchestration and studying interactions with other privacy mechanisms, such as differential privacy or secure inference, are promising avenues. Finally, applying this approach to larger foundation models and multimodal agent systems remains an important direction for future work.

\bibliography{example_paper}
\bibliographystyle{icml2026}
\clearpage
\appendix
\newpage
\section{Appendix}
\subsection{Proof of Theorem~\ref{thm:cumulative}}
\label{app:proof_cumulative}

We provide the full derivation of the cumulative leakage bound stated in
Theorem~\ref{thm:cumulative}. The proof relies on repeated applications of the
chain rule for mutual information (MI) and the conditional data processing inequality as given below.
\paragraph{Model Setup.}
We consider the sequential multi-agent system defined in
Section~3.1 and briefly recall the notation and assumptions required for the
theoretical analysis.

Let $S_i$ denote the agent-local sensitive variable accessible only to agent
$a_i$, $O_i$ the intermediate output produced by agent $a_i$, and $D_i$ the
task-relevant or public input available at stage $i$.
Agent outputs satisfy the update rule
$O_i = A_i(O_{i-1}, D_i, S_i)$, with $O_0 = \varnothing$.

We assume that the sensitive variables $\{S_i\}_{i=1}^N$ are mutually independent
and that the sequential pipeline satisfies the conditional Markov structure defined in Eq.~\eqref{eq:markov}.
Finally, each agent satisfies a local information leakage constraint as
\begin{equation}
I(O_i; S_i) \leq \epsilon_i .
\label{eq:local_privacy}
\end{equation}

\paragraph{Objective.}
Our goal is to upper bound the total compositional leakage:
\begin{equation}
I(O_N; S_1, \dots, S_N),
\end{equation}
in terms of the per-agent leakage budgets $\{\epsilon_i\}_{i=1}^N$.

\paragraph{Chain Rule Expansion.}

By the chain rule for MI:
\begin{equation}
I(O_N; S_1, \dots, S_N)
=
\sum_{i=1}^N I(O_N; S_i \mid S_{<i}).
\label{eq:chain_rule}
\end{equation}

\paragraph{Conditional Data Processing.}
\begin{definition}[Conditional Data Processing Inequality]
\label{def:cdpi}
If $X \rightarrow Y \rightarrow Z$ forms a Markov chain conditioned on $W$, then
\[
I(X; Z \mid W) \leq I(X; Y \mid W).
\]
\end{definition}

Applying Definition~\ref{def:cdpi} to the conditional Markov chain in
Eq.~\eqref{eq:markov}, we obtain:
\begin{equation}
I(O_N; S_i \mid S_{<i}) \leq I(O_i; S_i \mid S_{<i}).
\label{eq:cdpi_applied}
\end{equation}

Substituting into Eq.~\eqref{eq:chain_rule} yields:
\begin{equation}
I(O_N; S_1, \dots, S_N)
\leq
\sum_{i=1}^N I(O_i; S_i \mid S_{<i}).
\label{eq:reduced_sum}
\end{equation}
\paragraph{Conditional Leakage Decomposition.}

We define:
\begin{equation}
\mathcal{L}_{{i}} := I(O_i; S_i \mid S_{<i}).
\label{eq:def_ai}
\end{equation}

Since $S_i$ is independent of  $S_{<i}$, we have:
\begin{align}
\mathcal{L}_{{i}}
&= H(S_i \mid S_{<i}) - H(S_i \mid O_i, S_{<i}) \nonumber \\
&= H(S_i) - H(S_i \mid O_i, S_{<i}).
\end{align}

Similarly,
\begin{equation}
I(O_i; S_i) = H(S_i) - H(S_i \mid O_i).
\end{equation}

Subtracting the two expressions yields:
\begin{align}
\mathcal{L}_{{i}} - I(O_i; S_i)
&= H(S_i \mid O_i) - H(S_i \mid O_i, S_{<i}) \nonumber \\
&= I(S_i; S_{<i} \mid O_i).
\end{align}

Therefore,
\begin{equation}
\mathcal{L}_{{i}}
=
I(O_i; S_i)
+
I(S_i; S_{<i} \mid O_i).
\label{eq:ai_identity}
\end{equation}
\paragraph{Bounding Cross-Dependencies.}

\begin{lemma}[Upstream Leakage Bound]
\label{lem:upstream}
For each $i$, the following inequality holds:
\[
I(S_i; S_{<i} \mid O_i) \leq I(S_{<i}; O_{i-1}).
\]
\end{lemma}

\begin{proof}
Since $S_i$ is independent of $S_{<i}$, we have:
\[
I(S_i; S_{<i} \mid O_i)
=
I(S_{<i}; O_i \mid S_i) - I(S_{<i}; O_i).
\]
The second term is non-negative, yielding:
\[
I(S_i; S_{<i} \mid O_i)
\leq
I(S_{<i}; O_i \mid S_i).
\]
By the conditional Markov structure
$S_{<i} \rightarrow O_{i-1} \rightarrow O_i$ given $S_i$,
the data processing inequality implies:
\[
I(S_{<i}; O_i \mid S_i) \leq I(S_{<i}; O_{i-1} \mid S_i)
= I(S_{<i}; O_{i-1}),
\]
where the final equality follows from independence.
\end{proof}
\paragraph{Recursive Inequality.}

By Lemma~\ref{lem:upstream} and Eq.~\eqref{eq:ai_identity}:
\begin{equation}
\mathcal{L}_{{i}} \leq I(O_i; S_i) + I(S_{<i}; O_{i-1}).
\end{equation}

Using the chain rule:
\begin{equation}
I(S_{<i}; O_{i-1})
=
\sum_{j=1}^{i-1} I(S_j; O_{i-1} \mid S_{<j}).
\end{equation}

Applying conditional data processing again:
\[
I(S_j; O_{i-1} \mid S_{<j}) \leq I(S_j; O_j \mid S_{<j}) = \mathcal{L}_{{j}}.
\]

Thus,
\begin{equation}
\mathcal{L}_{{i}} \leq \epsilon_i + \sum_{j=1}^{i-1} \mathcal{L}_{{j}}.
\label{eq:recurrence}
\end{equation}
\paragraph{Solving the Recurrence.}

Define ${U}_{{k}} := \sum_{j=1}^k \mathcal{L}_{{j}}$. From Eqs. \eqref{eq:reduced_sum} and \eqref{eq:def_ai}, we note that 
\begin{align}\label{eq:upperboundmutinf}
I(O_N; S_1, \dots, S_N) \le U_N.    
\end{align}
From Eq.~\eqref{eq:recurrence}, we have:
\[
{U}_{{i}} = {U}_{{i-1}} + \mathcal{L}_{{i}} \leq \epsilon_i + 2{U}_{{i-1}}.
\]

Solving recursively yields:
\begin{align}
U_{{1}} &\leq \epsilon_1 \nonumber \\
U_{{2}} &\leq \epsilon_2 + 2\epsilon_1 \nonumber \\
U_{{3}} &\leq \epsilon_3 + 2\epsilon_2 + 4\epsilon_1 \nonumber \\
&\vdots \nonumber \\
U_{{N}} &\leq \sum_{i=1}^N 2^{N-i} \epsilon_i.
\end{align}
\paragraph{Conclusion.}

From Eq.~\eqref{eq:upperboundmutinf}, we obtain:
\begin{equation}
\boxed{
I(O_N; S_1, \dots, S_N)
\leq
\sum_{i=1}^N 2^{N-i} \epsilon_i
}
\end{equation}

In the uniform case $\epsilon_i = \epsilon$:
\[
I(O_N; S_1, \dots, S_N) \leq (2^N - 1)\epsilon.
\]

\subsection{Experimental Setup Details}
\label{app:exp_details}

\subsubsection{Tasks and Datasets}

\paragraph{MedQA (Medical Reasoning).}
We use the English subset of MedQA \citep{jin2020disease}, consisting of
multiple-choice clinical questions that require multi-step medical reasoning.
MedQA reflects high-stakes healthcare workflows, where intermediate reasoning
states may encode sensitive clinical attributes, implicit diagnoses, or
proprietary decision logic.

\paragraph{FinQA (Financial Reasoning).}
FinQA \citep{chen2021finqa} is a financial question-answering benchmark grounded
in real-world corporate reports.
The task requires numerical reasoning over tabular and textual evidence, making
it a privacy-sensitive setting where intermediate representations may encode
confidential financial indicators or internal analytical strategies.

\paragraph{PrivacyLens (Contextual Privacy Norms).}
PrivacyLens \citep{shao2024privacylens} evaluates privacy norm compliance in
action-based agent workflows.
Each data point is constructed hierarchically from a privacy-sensitive seed
(encoded as a five-tuple specifying data type, data subject, sender, recipient,
and transmission principle), which is expanded into a natural-language vignette
and an executable agent trajectory.
The trajectory consists of a sequence of agent actions and environment
observations that simulate a privacy-critical scenario, excluding the final
action.

In our experiments, agents consume the seed, vignette, and observed trajectory
and must generate the final action.
PrivacyLens evaluates whether this action leaks sensitive information specified
in the seed while maintaining task helpfulness, enabling realistic assessment of
sequential privacy leakage.
\subsubsection{Sequential Agent Architectures}

We evaluate sequential pipelines with $N \in \{2,3,4,5\}$ agents.
Each agent is an independent LLM with no shared parameters and communicates
exclusively through intermediate representations produced by upstream agents.

The agent roles are instantiated as follows:

\begin{itemize}
    \item \textbf{Agent $\boldsymbol{a_1}$} (\emph{Input Interpreter}): processes the user input
    (question or trajectory context) and produces an initial structured
    representation.
    \item \textbf{Agent $\boldsymbol{a_2}$} (\emph{Task Reasoner}): consumes upstream
    representations and performs core task-level reasoning or decision-making.
    \item \textbf{Agent $\boldsymbol{a_3}$} (\emph{Context Refiner}, when $N \ge 3$): refines and
    contextualizes intermediate outputs, potentially aggregating partial
    evidence or reformulating reasoning traces.
    \item \textbf{Agent $\boldsymbol{a_4}$} (\emph{Privacy-Aware Rewriter}, when $N \ge 4$):
    performs representation-level rewriting or filtering to reduce leakage of
    agent-local sensitive information while preserving task utility.
    \item \textbf{Agent $\boldsymbol{a_5}$} (\emph{Compliance Checker}, when $N = 5$): performs a
    final consistency and safety pass, ensuring that the output satisfies both
    task correctness and privacy constraints before release.
\end{itemize}

This design reflects realistic multi-agent deployments, where interpretation,
reasoning, refinement, and policy enforcement are handled by specialized
components operating sequentially. Importantly, increasing agent depth
introduces additional opportunities for compositional information leakage,
making this setting well-suited for studying privacy amplification effects.

\subsubsection{Models}

We evaluate two open-weight model families:
\begin{itemize}
    \item \textbf{Qwen}\footnote{
    \url{https://huggingface.co/Qwen/Qwen2-1.5B},
    \url{https://huggingface.co/Qwen/Qwen1.5-4B-Chat}
    }: 2B and 4B parameter variants
    \item \textbf{LLaMA}\footnote{
    \url{https://huggingface.co/meta-llama/Llama-3.2-3B},
    \url{https://huggingface.co/meta-llama/Llama-2-7b}
    }: 3B and 7B parameter variants
\end{itemize}

All agents are fine-tuned using LoRA \citep{hu2021lora} adapters with rank 8.
Unless stated otherwise, agents do not share LoRA parameters.

\subsubsection{Sensitive Variables}

Each agent $a_i$ is associated with an agent-local sensitive variable $S_i$,
which should not be inferable from its output $O_i$.

\begin{itemize}
    \item For MedQA and FinQA, $S_i$ includes embeddings of private prompts,
    internal routing states, or agent-specific policy indicators.
    \item For PrivacyLens, $S_i$ corresponds to privacy-sensitive seed elements,
    including data subject, recipient, and transmission principle.
\end{itemize}

\subsubsection{Training Objective}

The privacy-regularized objective is
\[
\mathcal{L}_{\text{total}}
=
\mathcal{L}_{\text{utility}}
+
\sum_{i=1}^{N} \beta_i \, \hat{I}_{\text{MINE}}(O_i; S_i),
\]
where $\hat{I}_{\text{MINE}}(O_i; S_i)$ is estimated using stabilized variational MI estimator (MINE). \citep{belghazi2018mine}
Privacy weights $\beta_i$ control the privacy-utility trade-off and are ablated
in Appendix~\ref{app:ablations}.

\subsubsection{Implementation Details}

All models are trained using AdamW \citep{loshchilov2019decoupled}  with learning rate $2\times10^{-5}$ and batch
size 16.
MI critics are two-layer MLPs with hidden dimension 256 and ReLU
activations.
All experiments are run with three random seeds and results are averaged.

\subsection{Detailed Evaluation Metrics}
\label{app:metrics}

This section provides formal definitions and implementation details for all
privacy and utility metrics used throughout the paper.
Our metrics are designed to jointly capture representation-level information
leakage, adversarial inference success, and task utility in sequential
multi-agent pipelines.
\subsubsection{Utility Metrics}

\paragraph{Cross-Entropy Loss (CE).}
For MedQA and FinQA, utility is primarily measured using the standard
CE loss between the model prediction $\hat{y}$ and the ground-truth
label $y$, where $x$ denotes the input query and $p_\theta(\cdot \mid x)$ is the
predictive distribution of the final agent parameterized by $\theta$:
\[
\mathrm{CE} = -\mathbb{E}_{(x,y)} \big[ \log p_{\theta}(\hat{y}=y \mid x) \big].
\]
Lower values indicate better task performance.

\paragraph{Benign Succeeded (BS).}
BS measures the fraction of benign (non-sensitive) inputs for
which the system produces a correct or helpful output, where
$N_{\text{benign}}$ denotes the number of benign evaluation instances and
$\mathbb{I}[\cdot]$ is the indicator function:
\[
\mathrm{BS}
=
\frac{1}{N_{\text{benign}}}
\sum_{i=1}^{N_{\text{benign}}}
\mathbb{I}[\hat{y}_i = y_i].
\]
For PrivacyLens, $\hat{y}_i = y_i$ is replaced by a helpfulness score
$\ge 2$ on the 0-3 rubric provided by the benchmark.
Higher values are better.

\subsubsection{Representation-Level Privacy Metrics}

\paragraph{Estimated Mutual Information.}
For each agent $a_i$, we estimate the MI between its output
representation $O_i$ and its agent-local sensitive variable $S_i$, where $O_i$
denotes the intermediate representation produced by agent $a_i$ and $S_i$
denotes the corresponding sensitive attribute accessible only at that stage:
\[
\hat{I}_{\text{MINE}}(O_i; S_i).
\]
MI is estimated using variational MI estimators
(MINE) \citep{belghazi2018mine}.
Lower MI indicates weaker statistical dependence and stronger privacy protection.

\paragraph{Average Mutual Information.}
To summarize privacy leakage across the depth of a sequential pipeline, we
report the average MI
\[
\mathrm{MI}_{\text{avg}} =
\frac{1}{N}\sum_{i=1}^{N}\hat{I}_{\text{MINE}}(O_i; S_i),
\]
where $N$ denotes the number of agents in the pipeline and
$\hat{I}_{\text{MINE}}(O_i; S_i)$ is the estimated MI at
agent $a_i$.
\subsubsection{Adversarial Privacy Metrics}

\paragraph{Leakage Accuracy (LA).}
LA measures the success rate of an adversary $\mathcal{A}$ that
attempts to infer sensitive attributes from agent outputs.
Let $O_i$ denote the output representation produced by agent $a_i$,
$S_i$ the corresponding ground-truth sensitive attribute, and
$\hat{S}_i = \mathcal{A}(O_i)$ the adversary’s prediction obtained by applying
$\mathcal{A}$ to the agent output:
\[
\mathrm{LA}
=
\frac{1}{N_{\text{sensitive}}}
\sum_{i=1}^{N_{\text{sensitive}}}
\mathbb{I}[\hat{S}_i = S_i],
\]
where $N_{\text{sensitive}}$ denotes the number of sensitive evaluation
instances and $\mathbb{I}[\cdot]$ is the indicator function.
Lower values indicate stronger privacy protection.

\paragraph{Sensitive Blocked (SB).}
SB measures the fraction of sensitive cases in which the
adversary fails to correctly infer the sensitive attribute:
\[
\mathrm{SB}
=
\frac{1}{N_{\text{sensitive}}}
\sum_{i=1}^{N_{\text{sensitive}}}
\mathbb{I}[\hat{S}_i \neq S_i]
= 1 - \mathrm{LA}.
\]
Higher values indicate stronger resistance to adversarial inference.

\subsubsection{Joint Privacy--Utility Metrics}

\paragraph{Balanced Outcome (BO).}
BO summarizes the privacy--utility trade-off by averaging benign
success and sensitive blocking:
\[
\mathrm{BO}
=
\frac{1}{2}(\mathrm{BS} + \mathrm{SB}).
\]
This metric penalizes methods that optimize privacy at the expense of utility,
or vice versa, by assigning equal weight to both objectives.

\paragraph{Overall Success (OS).}
OS measures whether a system simultaneously succeeds on benign
inputs and blocks sensitive leakage.
For paired benign/sensitive evaluation scenarios, let $N_{\text{pairs}}$ denote
the number of paired instances, $\text{BenignCorrect}_i \in \{0,1\}$ indicate
whether the benign instance is answered correctly, and
$\text{Leakage}_i \in \{0,1\}$ indicate whether sensitive information is leaked:
\[
\mathrm{OS}
=
\frac{1}{N_{\text{pairs}}}
\sum_{i=1}^{N_{\text{pairs}}}
\mathbb{I}[
\text{BenignCorrect}_i = 1
\;\wedge\;
\text{Leakage}_i = 0
].
\]
When paired evaluation is unavailable, we approximate OS as
\[
\mathrm{OS} \approx \mathrm{BS} \cdot \mathrm{SB},
\]
which assumes independence between benign correctness and leakage events.

\subsubsection{PrivacyLens-Specific Metric Mapping}

PrivacyLens evaluates action-based privacy leakage using Leakage Rate (LR) and
Adjusted Leakage Rate (LRh), computed over actions with positive helpfulness.
We map these quantities to our metrics as follows:
\begin{itemize}
    \item \textbf{Leakage Accuracy (LA)} := $\mathrm{LRh}$
    \item \textbf{Sensitive Blocked (SB)} := $1 - \mathrm{LRh}$
    \item \textbf{Benign Succeeded (BS)} := fraction of actions with helpfulness $\ge 2$
    \item \textbf{Overall Success (OS)} := $\mathrm{BS} \cdot \mathrm{SB}$
\end{itemize}
This mapping allows consistent comparison between representation-level MI
metrics and action-based privacy outcomes.

\subsubsection{Composite Metric: Privacy-Aware Reasoning Index (PARI)}

To summarize privacy--utility trade-offs in a single scalar metric, analogous to
the Responsible Reasoning Index (RRI) introduced by \citet{raza2025responsible},
we define the \emph{Privacy-Aware Reasoning Index (PARI)} as:
\[
\mathrm{PARI}
=
w_1 \cdot RC
+ w_2 \cdot OT
+ w_3 \cdot PI
+ w_4 \cdot RT.
\]

PARI is intended as a system-level summary metric that aggregates multiple
existing privacy and utility signals, rather than introducing new primitive
measurements.

\paragraph{Reasoning Correctness (RC).}
We define reasoning correctness directly in terms of benign task success using
the Benign Succeeded (BS) metric introduced earlier:
\[
RC := \mathrm{BS}.
\]
This choice avoids introducing a separate correctness score and ensures direct
consistency between task-level utility evaluation and the composite metric.
\begin{figure}[H]
\vspace{-6pt}
\centering
\includegraphics[width=0.85\columnwidth]{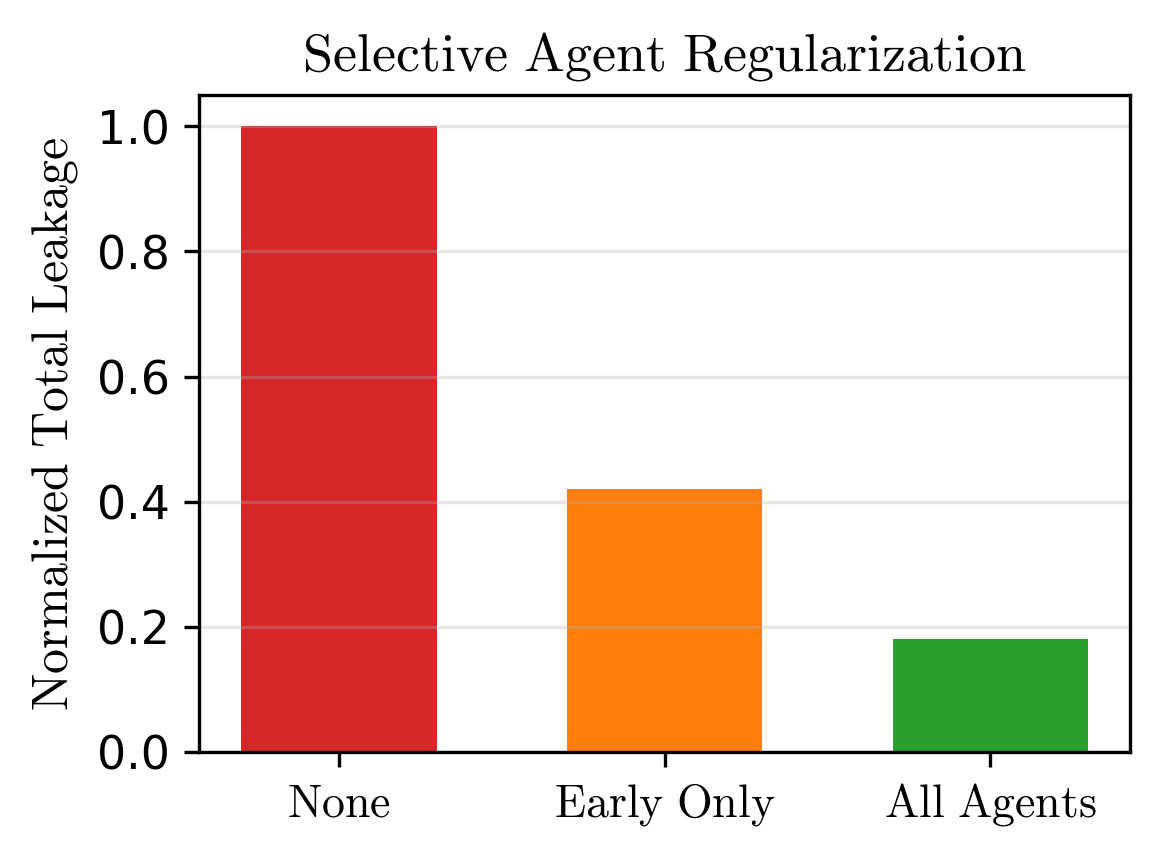}
\vspace{-8pt}
\caption{Total compositional leakage under selective agent regularization.}
\label{fig:selective}
\end{figure}

\paragraph{Operational Transparency (OT).}
Operational Transparency measures the stability of intermediate representations
under benign input perturbations.
Let $q$ denote a benign input query, $q+\epsilon$ a semantically equivalent
perturbation, and $O(\cdot)$ the final system representation produced by the
agent pipeline.
We define:
\[
OT := 1 - \mathbb{E}\big[\|O(q) - O(q+\epsilon)\|_2\big].
\]
In our experiments, representation stability remains consistent across methods;
thus, OT is treated as constant and does not affect relative comparisons.

\paragraph{Privacy Integrity (PI).}
Privacy Integrity quantifies the relative reduction in representation-level
leakage.
Let $\mathrm{MI}_{\text{avg}}$ denote the average estimated MI
across agents for a given method, and $\mathrm{MI}_{\text{baseline}}$ the
corresponding value for the unregularized baseline.
We define:
\[
PI := 1 - \frac{\mathrm{MI}_{\text{avg}}}{\mathrm{MI}_{\text{baseline}}}.
\]

\paragraph{Reproducibility and Traceability (RT).}
Reproducibility and Traceability capture whether system execution is deterministic
and auditable.
Since all experiments enforce fixed routing, deterministic decoding, and
comprehensive logging, this component is fixed:
\[
RT := 1.
\]

\paragraph{Interpretation and Metric Mapping.}
PARI does not introduce new evaluation primitives; rather, it aggregates
previously defined metrics into a single scalar score.
Specifically, RC corresponds directly to BS, PI is a normalized
function of average representation-level MI, OT captures representation stability under benign
perturbations, and RT encodes deterministic execution.
This decomposition allows PARI to summarize the privacy--utility trade-off while
preserving interpretability in terms of underlying metrics.

\paragraph{Weights.}
Unless otherwise stated, we use:
\[
(w_1, w_2, w_3, w_4) = (0.3, 0.1, 0.5, 0.1),
\]
reflecting the central importance of privacy integrity in our evaluation.

\subsection{Additional Experiments on PrivacyLens}
\label{app:add-exp}

We provide additional validation of our framework on the
\textbf{PrivacyLens} benchmark~\citep{shao2024privacylens},
which is specifically designed to evaluate privacy leakage and
adversarial inference risks in language models.
Unlike MedQA and FinQA, PrivacyLens does not emphasize task-level reasoning,
but instead probes the model’s ability to resist sensitive attribute inference
under controlled adversarial settings.

To isolate privacy behavior in resource-constrained regimes, we report results
only for \textbf{Qwen-2B} and \textbf{LLaMA-3B}.
All evaluation metrics follow the same definitions as in
Appendix~\ref{app:metrics}, and all MI estimates are computed
using frozen estimators.
\subsubsection{PrivacyLens Results}

\paragraph{Discussion.}
The PrivacyLens results reinforce the core findings of the main paper.
Across both Qwen-2B and LLaMA-3B, MI-regularized training reduces $\mathrm{MI}_{\text{avg}}$ by approximately \textbf{65--80\%} relative to unregularized
baselines, leading to substantial improvements in resistance to sensitive
attribute inference.
SB increases consistently across agent depths, while
LA decreases in near-linear correspondence with $\mathrm{MI}_{\text{avg}}$.

As expected, absolute benign performance (BS) is slightly lower than in
MedQA and FinQA, reflecting the adversarial nature of PrivacyLens.
Nevertheless, utility degradation remains moderate, and the resulting
OS improves by more than an order of magnitude in most
settings.
These results confirm that the benefits of MI-regularization extend beyond
task-centric benchmarks to dedicated privacy stress tests.

Importantly, the observed trends mirror those in the main results:
unregularized systems exhibit rapid leakage amplification with agent depth,
while MI-regularization effectively suppresses cumulative leakage.
This consistency across benchmarks strengthens the case for MI
as a principled, system-level control mechanism for privacy in sequential
multi-agent LLM deployments.

\subsection{Ablation Studies}
\label{app:ablations}

We conduct a series of ablation studies to analyze the behavior of the proposed
privacy-regularized training framework and to validate key design choices.
Unless otherwise stated, all ablations are evaluated using validation accuracy
as the utility metric and estimated MI leakage
$\hat{I}_{\text{MINE}}(O_i; S_i)$ as the privacy metric.

\subsubsection{Effect of Privacy Weight $\beta$}

We first examine the effect of the privacy regularization strength by sweeping
$\beta \in \{0, 10^{-3}, 10^{-2}, 10^{-1}, 1\}$ with uniform allocation
$\beta_1 = \beta_2 = \beta$. This experiment isolates the role of $\beta$ as a
global control parameter governing the privacy-utility trade-off.

\begin{table}[t]
\centering
\caption{PrivacyLens results for lightweight models
(B = Baseline, M = MINE-Reg).}
\label{tab:privacylens_results}
\scriptsize
\setlength{\tabcolsep}{4pt}
\begin{tabular}{@{}l c @{} c c c c c c c @{}}
\toprule
\textbf{Model} & \textbf{Ag}\kern3pt & \textbf{Meth} &
\textbf{MI$_{\text{avg}}$} $\downarrow$ &
\textbf{LA} $\downarrow$ &
\textbf{SB} $\uparrow$ &
\textbf{BS} $\uparrow$ &
\textbf{OS} $\uparrow$ &
\textbf{PARI} $\uparrow$ \\
\midrule
Qwen-2B & 2 & B & 0.64 & 0.61 & 0.39 & 0.90 & 0.12 & 0.468 \\
Qwen-2B & 2 & M & 0.12 & 0.25 & 0.75 & 0.86 & 0.65 & 0.862 \\
Qwen-2B & 3 & B & 0.82 & 0.68 & 0.32 & 0.87 & 0.07 & 0.459 \\
Qwen-2B & 3 & M & 0.17 & 0.29 & 0.71 & 0.80 & 0.57 & 0.846 \\
Qwen-2B & 4 & B & 0.98 & 0.73 & 0.27 & 0.84 & 0.04 & 0.451 \\
Qwen-2B & 4 & M & 0.22 & 0.34 & 0.66 & 0.77 & 0.51 & 0.832 \\
\midrule
LLaMA-3B & 2 & B & 0.69 & 0.66 & 0.34 & 0.86 & 0.10 & 0.453 \\
LLaMA-3B & 2 & M & 0.21 & 0.37 & 0.63 & 0.78 & 0.49 & 0.781 \\
LLaMA-3B & 3 & B & 0.87 & 0.71 & 0.29 & 0.83 & 0.06 & 0.444 \\
LLaMA-3B & 3 & M & 0.26 & 0.41 & 0.59 & 0.75 & 0.44 & 0.776 \\
LLaMA-3B & 4 & B & 1.02 & 0.75 & 0.25 & 0.80 & 0.04 & 0.435 \\
LLaMA-3B & 4 & M & 0.31 & 0.45 & 0.55 & 0.72 & 0.40 & 0.764 \\
\bottomrule
\end{tabular}
\end{table}

\begin{figure}[H]
\vspace{-6pt}
\centering
\includegraphics[width=0.85\columnwidth]{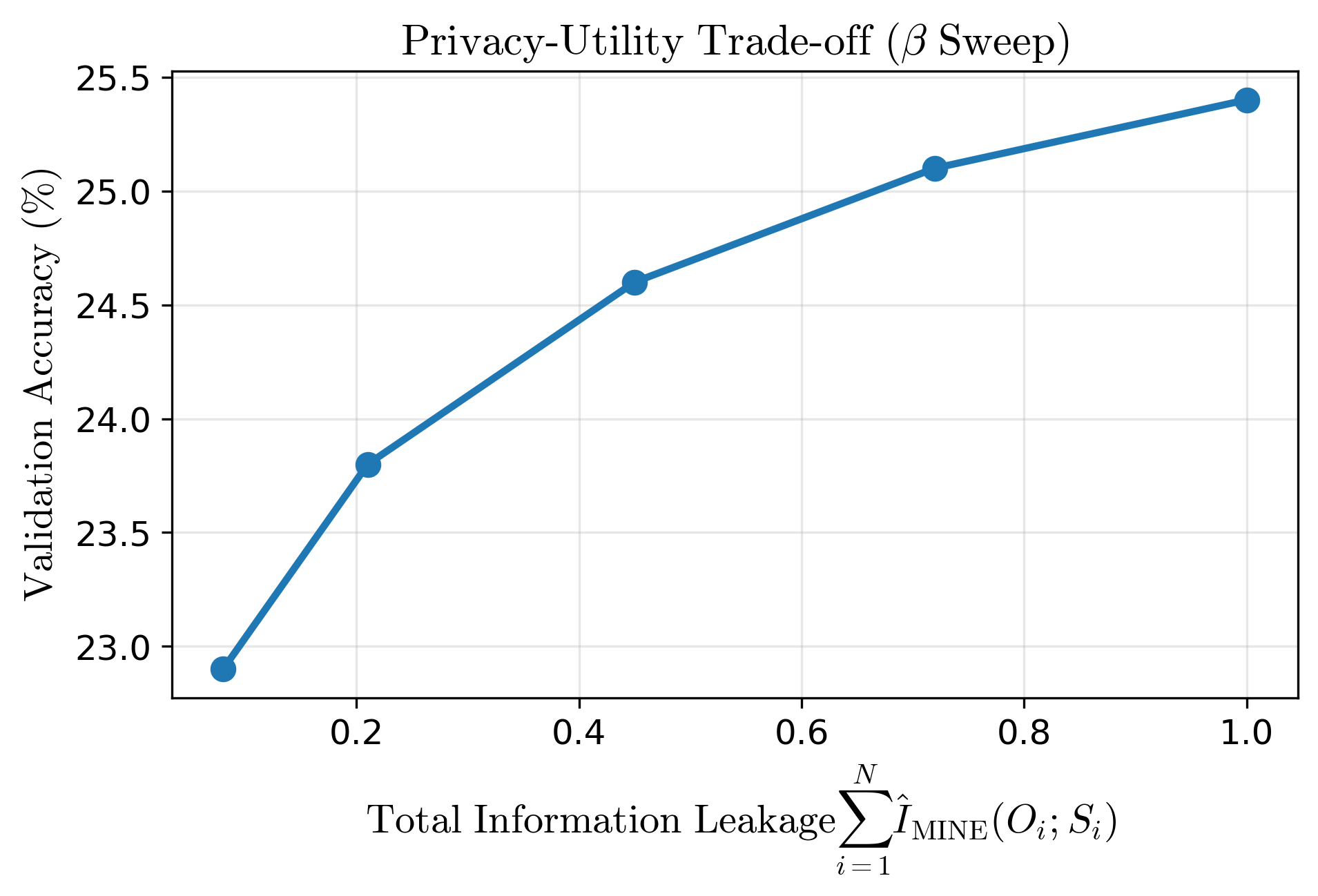}
\vspace{-8pt}
\caption{
Privacy-utility trade-off as a function of the privacy weight $\beta$,
measured by validation accuracy versus total information leakage
$\sum_{i=1}^{N} \hat{I}_{\text{MINE}}(O_i; S_i)$,
for $N{=}3$ sequential-agent pipeline using Qwen-2B evaluated on MedQA.
}

\label{fig:beta_sweep}
\end{figure}

Figure~\ref{fig:beta_sweep} reveals a smooth and monotonic trade-off between
privacy and utility. As $\beta$ increases, total information leakage decreases
consistently, while validation accuracy degrades gradually rather than
catastrophically. Notably, the utility curve exhibits diminishing sensitivity
beyond moderate leakage reductions, indicating that a substantial fraction of
sensitive information can be removed before task-relevant performance is
severely affected.

This behavior suggests that $\beta$ functions as a stable and interpretable
privacy knob, enabling practitioners to select operating points that achieve
meaningful privacy gains without incurring abrupt performance collapse. The
absence of sharp phase transitions further indicates that the optimization
landscape remains well-conditioned across a wide range of privacy budgets.

\subsubsection{Selective Agent Regularization}

We study whether privacy regularization must be applied to all agents or only
to a subset of the pipeline. We compare three settings:
(i) no agents regularized,
(ii) only early agents regularized, and
(iii) all agents regularized.
All settings use the same total privacy budget.

Figure~\ref{fig:selective} shows that regularizing only early agents substantially
reduces total leakage relative to the unregularized baseline, highlighting the
importance of controlling information release at early stages. However,
significant residual leakage persists due to amplification in later agents that
remain unconstrained.

Full regularization across all agents consistently achieves the lowest leakage,
demonstrating that privacy risks are distributed throughout the pipeline rather
than localized to initial stages. This result emphasizes the need for
system-level privacy control and cautions against partial defenses that leave
downstream agents unregulated.

\subsubsection{Representation Compression Effects}

We analyze how MI regularization affects the information content
of intermediate representations. Specifically, we measure
$I(O_i; y)$ (task relevance) and $I(O_i; S_i)$ (privacy leakage) as $\beta$
increases.

\begin{figure}[H]
\vspace{-6pt}
\centering
\includegraphics[width=0.85\columnwidth]{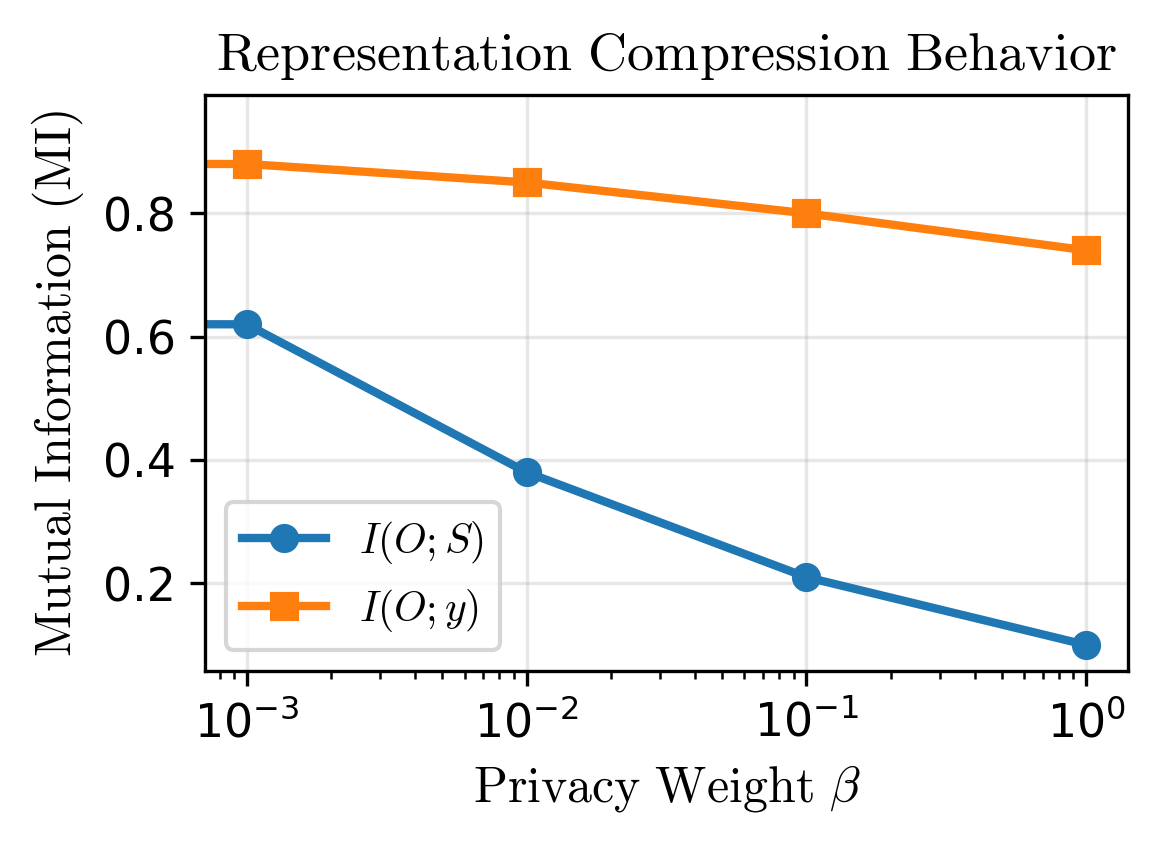}
\vspace{-8pt}
\caption{Effect of privacy regularization on task-relevant and sensitive
information in intermediate representations.}
\label{fig:repr_tradeoff}
\end{figure}

Figure~\ref{fig:repr_tradeoff} illustrates that increasing $\beta$ leads to a
pronounced reduction in $I(O_i; S_i)$, while $I(O_i; y)$ decreases much more
gradually. This asymmetric compression indicates that the regularizer
preferentially suppresses sensitive information rather than indiscriminately
collapsing representations.

The preservation of task-relevant signal explains the smooth utility degradation
observed in earlier experiments and suggests that the model learns to discard
privacy-sensitive features while retaining predictive structure. This behavior
aligns with the intended role of information-theoretic regularization as a
selective bottleneck rather than a blunt constraint.

\subsubsection{Robustness to Adversarial Inference}

Finally, we evaluate whether reductions in MI translate into
practical robustness against inference attacks. We train a lightweight
adversarial probe to predict the sensitive attribute $S_i$ from agent outputs
$O_i$ under varying privacy budgets.

\begin{figure}[H]
\vspace{-6pt}
\centering
\includegraphics[width=0.85\columnwidth]{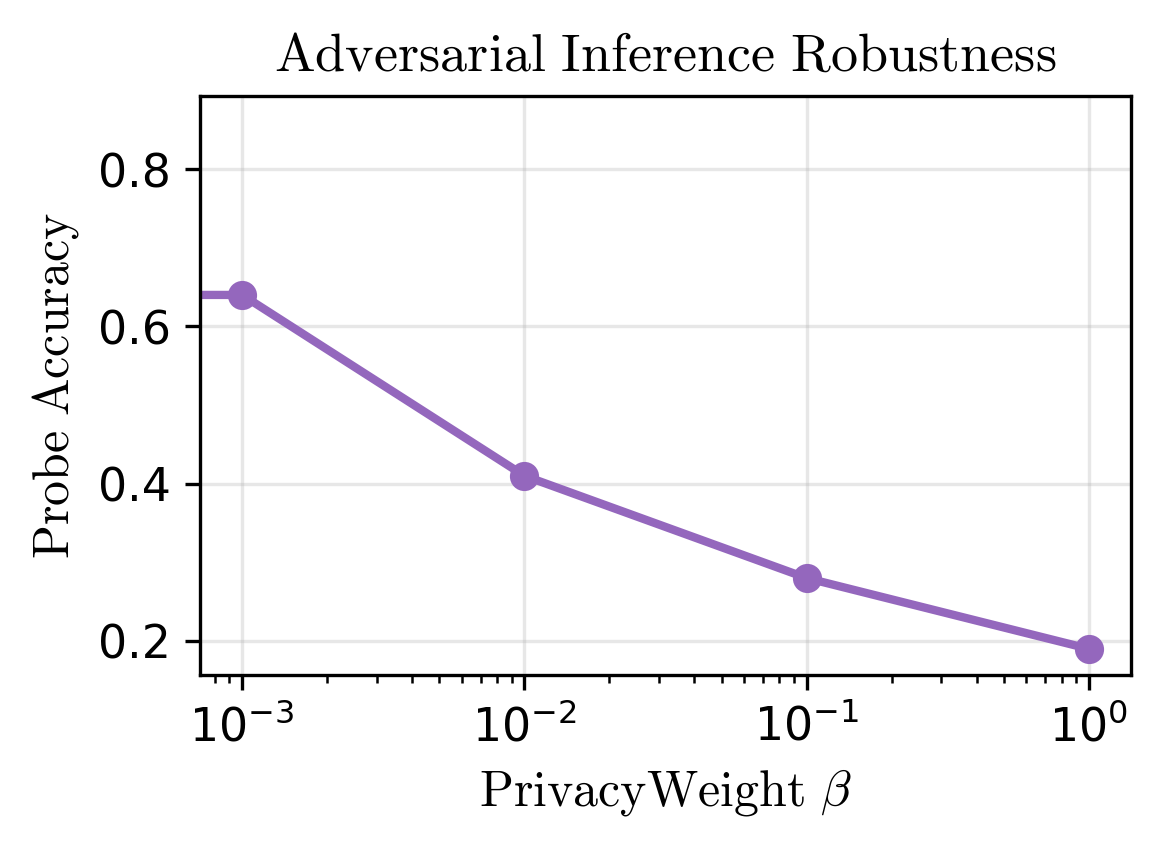}
\vspace{-8pt}
\caption{Adversarial probe accuracy for predicting sensitive variables from
agent outputs.}
\label{fig:probe}
\end{figure}

As shown in Figure~\ref{fig:probe}, adversarial probe accuracy decreases
monotonically as $\beta$ increases, closely tracking the reduction in estimated
MI. At higher privacy weights, probe performance approaches
chance level, indicating that sensitive attributes are no longer reliably
recoverable from intermediate representations.

This result confirms that MI regularization yields operational privacy benefits
beyond metric-level improvements, directly limiting the ability of downstream
attackers to exploit latent representations for sensitive inference.

\subsection{Extended Related Work}
\label{app:related_work_extended}

\subsubsection{Memorization and Training-Data Privacy}

Early work on LLM privacy focused on memorization-based risks, demonstrating that
models can reproduce rare or repeated training examples verbatim
\citep{xu2025positional, nasr2025scalable}.
Membership inference attacks further show that adversaries can determine whether
a specific record was used during training
\citep{shokri2017membership}.
Subsequent work explores mitigating these risks using differential privacy during
training \citep{li-etal-2025-1, kulynych2025unifying}.
While these methods provide strong record-level guarantees, they operate at
training time and do not address inference-time leakage arising from contextual
reasoning or interaction history.

\subsubsection{Inference-Time Privacy and User Profiling}

Beyond memorization, inference-time privacy attacks infer sensitive attributes
such as demographics, preferences, or intentions from user prompts
\citep{patil2025sum, zhang2025effective}.
User profiling literature similarly studies attribute inference from text
\citep{ yang2025agentnet,gomaa2025converse}.
These attacks typically assume a single model interacting with a single user,
and do not capture risks that emerge when partial information is distributed
across multiple agents.

\subsubsection{Contextual Integrity and Privacy Norms}

Contextual integrity theory posits that privacy violations arise when information
flows deviate from social norms defined by actors, information type, and
transmission principles \citep{gomaa2025converse, rosser2025agentbreeder}.
Recent work applies this framework to LLMs, emphasizing norm-aware decision making
in dialogue systems.
PrivacyLens operationalizes these ideas by constructing privacy-sensitive seeds,
vignettes, and trajectories for action-based evaluation
\citep{shao2024privacylens}.
However, these evaluations focus on final actions and do not quantify how
intermediate representations accumulate privacy risk across agent pipelines.

\subsubsection{Multi-Agent Systems and Memory Risks}

LLM-based multi-agent systems have gained traction for collaborative reasoning,
tool use, and autonomous workflows
\citep{talebirad2023multiagent, khan2025generative}.
A critical component of these systems is shared or persistent memory, which
stores intermediate reasoning traces and observations.
Recent work highlights that agent memory can be exploited through poisoning or
indirect manipulation, enabling targeted misbehavior without direct access to
the memory store \citep{chen2024agentpoison, srivastava2025memorygraft}.
Other studies show that agent topologies and communication patterns introduce
unique vulnerabilities absent in single-agent systems.

Despite these advances, existing defenses largely rely on heuristic filtering,
memory sanitization, or access control, and lack formal guarantees against
privacy leakage arising from sequential composition.
Our work complements this literature by providing a principled, representation-
level analysis of privacy leakage and a training-time mitigation strategy that
scales with agent depth.

\end{document}